\documentclass[a4paper]{article}
\usepackage{graphicx}
\usepackage{geometry}
\usepackage{amssymb}
\usepackage{amstext}
\usepackage{amsmath}
\usepackage{amsthm}
\usepackage{multicol}
\usepackage{mathrsfs}
\usepackage{mathtools}
\usepackage{nccmath}
\usepackage{tikz}
\usepackage{here}
\usepackage{caption}
\usepackage{setspace}
\usepackage{authblk}
\usepackage{fancybox}
\usepackage{dashbox}
\usepackage{enumitem}
\usepackage{mathcomp}
\usepackage{bussproofs}
\usepackage{url}

\usetikzlibrary{positioning}
\usetikzlibrary{arrows.meta}

\theoremstyle{plain}
\newtheorem{df}{Definition}
\newtheorem{thm}{Theorem}

\newtheorem{exa}{Example}

\newtheorem{lemma}{Lemma}

\title{A Representation of Explicit Knowledge and Epistemic Indistinguishability in a Logic of Awareness}
\author{Yudai Kubono$^1$ and Satoshi Tojo$^2$}  
\affil{$^1$\textit{Graduate School of Science and Technology, Shizuoka University,} \par \textit{Ohya, Shizuoka 422-8529, Japan} \par \protect\url{yudai.kubono@gmail.com}}
\affil{$^2$\textit{Department of Data Science, Asia University,} \par \textit{Sakai, Tokyo 180-8629, Japan} \par \protect\url{tojo_satoshi@asia-u.ac.jp}}

\date{\today}  

\begin{document}
\maketitle
\noindent
The final version of this paper will appear in the \textit{Journal of Applied Logics} (\url{https://www.collegepublications.co.uk/ifcolog/}).
% \vspace{-6pt}
\begin{abstract}
The logic of awareness, first proposed by Fagin and Halpern, addressed the problem of logical omniscience by introducing the notion of awareness and distinguishing explicit knowledge from implicit knowledge. In their framework, explicit knowledge was defined as the conjunction of implicit knowledge and awareness, each of which was represented by modal operators. Their definition, however, may derive undesirable propositions that cannot be considered explicit knowledge when Modus Ponens is applied within implicit knowledge. Hence, focusing on indistinguishability among possible worlds, dependent on awareness, we refine the definition of explicit knowledge. In our semantics, we require that the aware implicit knowledge is not necessarily explicit knowledge, though explicit knowledge must be aware as well as implicit. We employ an example of elementary geometry, where different students may or may not reach the final answer, depending on whether they are aware of learned mathematical facts. Thereafter, we formally present the syntax and the semantics of our language, named Awareness-Based Indistinguishability Logic ($\mathrm{AIL}$). We prove that $\mathrm{AIL}$ has more expressive power than the logic of Fagin and Halpern, and show that the latter is embeddable in $\mathrm{AIL}$. Furthermore, we provide an axiomatic system of $\mathrm{AIL}$ and prove its soundness and completeness.
\end{abstract}

\textit{Keywords}: Awareness, Awareness logic, Explicit knowledge, Epistemic logic.

\section{Introduction}

Epistemic Logic (henceforth, $\mathrm{EL}$), as a variant of modal logic, employs the modal operator $K$,  where $K\varphi$ means ``an agent knows $\varphi$.'' However, this $K$ operator causes the well-known problem of \textit{logical omniscience}. For instance, the axioms and inference rules of $\mathrm{EL}$ derive that an agent has knowledge of all tautologies and that when ``an agent knows $\varphi$'' and ``the agent knows $\varphi \to \psi$,'' the agent always knows $\psi$. Many approaches have been proposed to close this gap \cite{fagin1995reasoning}. 
The logic of awareness (henceforth, FH logic\footnote{The logical semantics proposed in \cite{fagin1988belief}, an \textit{awareness structure}, can be given some restrictions, and we can consider several classes of the structures and the corresponding logics.}) \cite{fagin1988belief} addressed the issue of logical omniscience, incorporating the notion of \textit{awareness} to distinguish \textit{implicit}/\textit{explicit} knowledge, where the implicit knowledge behaves as the usual $K$ in $\mathrm{EL}$ while the explicit knowledge is the \textit{aware} implicit knowledge. Even if $\varphi$ is a tautology, if it is not aware by the agent, it does not belong to her explicit knowledge. This concise and intuitive idea has become a mainstay in the literature on logics of awareness \cite{van2015handbook}.

Their definition, however, may derive undesirable propositions that cannot be considered explicit knowledge, when \textit{Modus Ponens}\footnote{From $\varphi$ and $\varphi\to\psi$, $\psi$ is concluded.} ($\mathrm{MP}$) is applied within implicit knowledge. Let $I, A$, and $E$ represent the operators for implicit knowledge, awareness, and explicit knowledge, respectively. We primarily focus on a valid formula in FH logic, $I(\varphi) \wedge I(\varphi \to \psi) \to (A\psi \to E\psi)$, which is entailed from (1): logical omniscience for implicit knowledge $I(\varphi) \wedge I(\varphi \to \psi) \to I\psi$ and (2): the conjunctive definition of explicit knowledge $I\psi \wedge A\psi \leftrightarrow E\psi$. This formula states that the inferred consequent is classified as explicit knowledge whenever the agent is aware of it, even if not aware of the antecedents themselves. On the other hand, one may take the view that when an agent explicitly claims to know a proposition that is derived through inference, it must be derived solely from propositions of which the agent is aware \cite{bergmann2006justification}. On this view, implicit knowledge conjunct with awareness need not suffice for explicit knowledge. 
Hence, this paper revisits the conjunctive definition of explicit knowledge, refining its representation so that $E\varphi\to I\varphi\wedge A\varphi$ but $I\varphi\wedge A\varphi\not\to E\varphi$. The following section highlights the existing gap between the formal representation of the knowledge and its intended interpretation.

Traditionally, there has been a distinction between \textit{awareness of} and \textit{awareness that}; the former restricts the \textit{aware} atomic propositions and is dealt in economics \cite{heifetza2008canonical,halpern2008interactive}, philosophy \cite{grossi2015syntactic,Fernandez-Fernandez2021-FERAIL-2}, and computer science \cite{aagotnes2014logic,lorini2020grounding}, while the latter concerns complex formulae \cite{grossi2015syntactic,fernandez2021awareness,Fernandez-Fernandez2021-FERAIL-2} and furthermore the computable truth with time restriction \cite{fagin1988belief,van2010dynamics}. In this paper, we focus only on \textit{awareness of}, to elucidate the revision from the preceding works.

The rest of this paper is structured as follows. In Section 2, we provide a motivating example and propose a semantic notion to refine the FH logic's definition of explicit knowledge. In Section 3, we formally define the syntax and semantics of our language, introducing a logic named \textit{Awareness-Based Indistinguishability Logic} ($\mathcal{AIL}$). We use $\mathcal{AIL}$ to refer to the set of formulae semantically characterized, i.e., the set of all valid formulae in the class of our models. Thereafter, we demonstrate that our formalization faithfully captures the knowledge in the example introduced in Section 2. In Section 4, we prove that our logic is more expressive than an FH logic, which means that for every formula of the FH logic, there exists a formula in the language of $\mathcal{AIL}$ that has the same meaning, and show that the latter is embeddable into $\mathcal{AIL}$. In Section 5, we provide an axiomatic system \textbf{AIL} in the Hilbert style and prove its soundness and completeness. Section 6 is devoted to related work. Finally, Section 7 concludes with directions for future research.

\section{Examples with Awareness and Indistinguishability}

Before introducing a motivating example, we clarify the notions of awareness and knowledge focused on in this paper.
\begin{description}
  \item[Awareness] This notion concerns which information is available to the agent at a given moment and not the acknowledgement of its truth or falsity; thus, ``to be aware of a proposition'' means ``to be able to refer to a proposition'' or ``to entertain a proposition.''
  \item[Implicit Knowledge] This comprises what an agent has independently of the \allowbreak agent's current awareness. Furthermore, the knowledge closes under logical consequences and includes all tautologies, thus constituting the ideal knowledge of a logically omniscient agent.
  \item[Explicit Knowledge] This comprises what an agent has in the agent's awareness and hence is available for reasoning or decision-making at that moment. This knowledge includes the aspect of availability expressed by the notion of awareness, referring to what the agent actually knows at a given time. 
\end{description}
\noindent
The following example succinctly shows the difference between those notions of knowledge:
\begin{exa}
 An agent knows the color of her smartphone, but it is not at the forefront of her mind, say, while she is reading a cookbook. 
\end{exa}
\noindent
The implicit knowledge becomes available to her when attention is drawn to it, for instance, when someone specifically asks about it.

Since the definition of awareness is given as a subset of propositions in a possible world, the explicit knowledge is only a limitation or a filtered one that originally resides in each world. Therefore, Examples 2 and 3 are inadequate in this context.
\begin{exa}
  One day, you are introduced to a cousin who has never met before. Furthermore, you have not known that you have had such a relative.
\end{exa}
\noindent
Although it is true that he or she is the cousin of ``you'', this cannot be implicit knowledge since the information does not reside in each world.

\begin{exa}
  You are playing a chess game, fully familiar with the rules, and there is an action to checkmate; however, you do not notice it.
\end{exa}
\noindent
In this case, there may be hidden propositions of which ``you'' are not aware; however, even though these atomic propositions are externalized, ``you'' may not be able to find the path to checkmate because of so large a search space. This raises issues of resource-boundedness tied to awareness and, in turn, to explicit knowledge. The notion of awareness discussed in this paper does not aim to capture this kind of case, and we do not pursue such resource-boundedness here. For studies addressing reasoning on explicit belief/knowledge, see \cite{balbiani2016logical,Fernandez-Fernandez2021-FERAIL-2}. Nor do we address this kind of dynamics.

\subsection{A Motivating Example}

\begin{figure}
  \centering
  % \hspace{-2zh}
  \includegraphics[width=7.6cm]{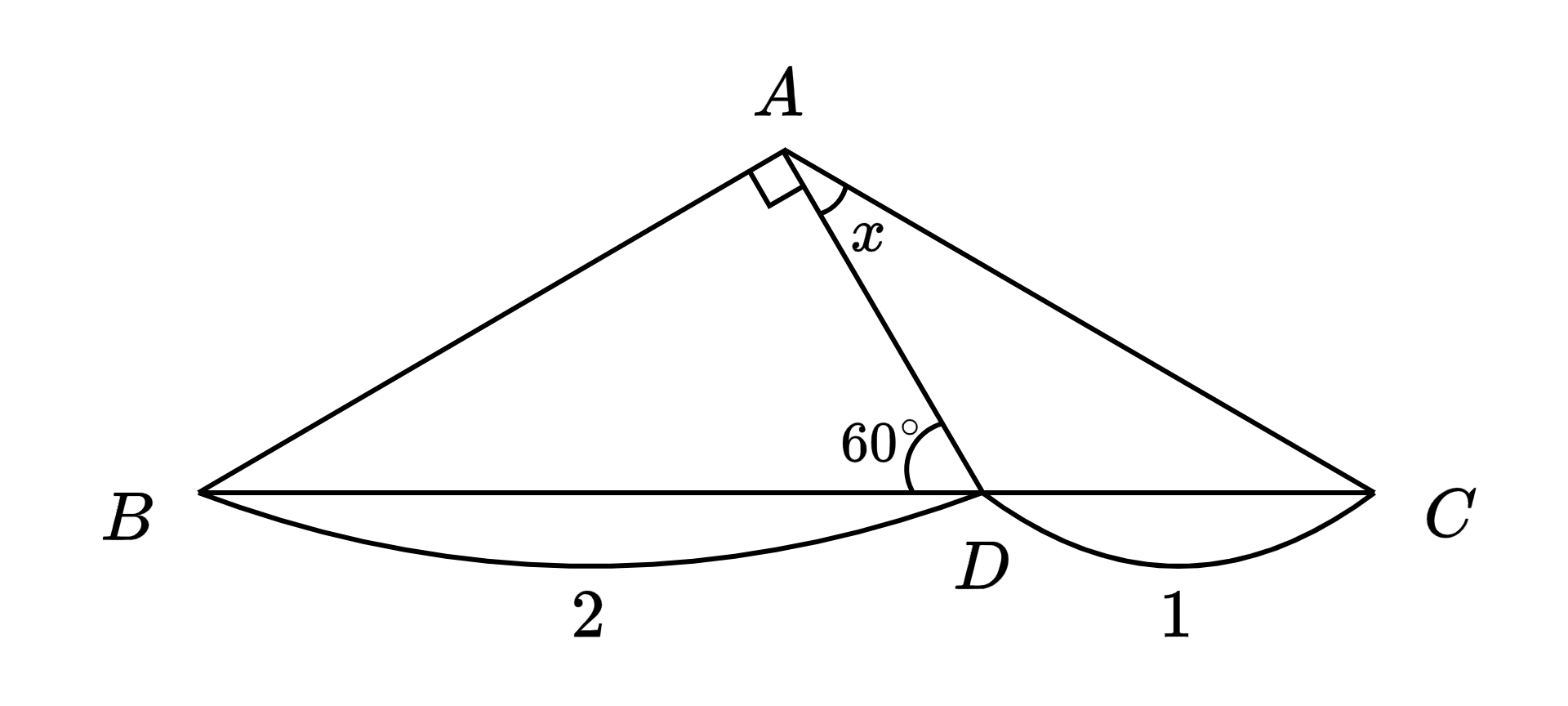}
  \caption{An elementary geometry problem.}
\end{figure}
\noindent

We introduce an elementary geometry problem as a motivating example.
The problem illustrated in Figure 1 can be solved by utilizing mathematical facts, such as those related to the ratios in a $30\tcdegree$-$60\tcdegree$-$90\tcdegree$ triangle and the measure of an exterior angle. 
Specifically, let $f_i$ be the mathematical facts, that is: 
\begin{itemize}
    \item $f_1$: A triangle with angles $30\tcdegree, 60\tcdegree, 90\tcdegree$ has sides in the ratio $1\colon\sqrt{3}\colon 2$ (opposite $30\tcdegree, 60\tcdegree, 90\tcdegree$, respectively);
    \item $f_2$: The measure of an exterior angle is equal to the sum of the measures of the two remote (non-adjacent) interior angles of the triangle;
    \item $f_3$: The measures of angles opposite the equal sides are equal, 
\end{itemize}
$k_i$ given knowledge for $i\in\{1,2,3\}$, and $p_i$ obtained knowledge for $i\in\{1,2,3,4\}$; then the proof in Figure 2 can be constructed.

\begin{figure}
  \centering
\begin{prooftree}
    \AxiomC{$f_1$}
    \AxiomC{$k_1\colon \triangle BDA \text{ is }30\tcdegree$-$60\tcdegree$-$90\tcdegree$ triangle}
    \AxiomC{$k_2\colon BD=2$}
    \TrinaryInfC{$p_1\colon DA=1$}
    \AxiomC{$k_3\colon DC=1$}
    \BinaryInfC{$p_2\colon DA = DC$}
\end{prooftree}
\begin{prooftree}
    \AxiomC{$f_2$}
    \AxiomC{$k_1$}
    \RightLabel{($\dag$)}
    \BinaryInfC{$p_3\colon \angle DAC + \angle DCA = 60\tcdegree$}
    \AxiomC{$p_2$}
    \AxiomC{$f_3$}
    \RightLabel{($\ddag$)}
    \TrinaryInfC{$p_4\colon \angle DAC = 30\tcdegree$}
\end{prooftree}
  \caption{A proof of Figure 1. Note that the given hypotheses or derived facts are shown in the upper row of each horizontal line, and their immediate consequence is shown in its lower row, though the proof format is not based on a known formal system.}
\end{figure} 

We consider two students who differ in awareness and whether each can correctly respond to the question\footnote{The questions ``What is the measure of angle $x$?'' and ``Is angle $x$ $30$ degrees?'' are different. The former suggests no specific value; its declarative form is not a proposition. In contrast, the declarative sentence corresponding to the latter, ``Angle $x$ is $30$ degrees.'', is a proposition. In this running example, agents are required to respond to the question about whether this proposition is true or false simply with ``yes'' or ``no.''}: ``Is angle $x$ $30$ degrees ($p_4$)?''
\begin{exa}
Two students, $a$ and $b$, take a geometry exam and answer the question: ``Is angle $x$ $30$ degrees?'' We assume that their implicit knowledge includes $f_i$\footnote{$f_1$, $f_2$, and $f_3$ are theorems of a geometric theory but not logical truths. In this paper, we do not assume a theory such as Euclidean geometry; therefore, those mathematical facts are not valid formulae.} and further the implications (e.g., $(f_1 \wedge k_1 \wedge k_2)\to p_1$) connecting the upper propositions to the lower one at each inference step in Figure 2. Student $a$ is aware of each $f_i$, whereas $b$ is aware only of $f_1$ and $f_3$. Since $a$ is aware of all the mathematical facts, she knows $p_4$ explicitly. In contrast, $b$ is unaware of $f_2$ and thus unable to derive $p_3$ and $p_4$, which means she cannot respond correctly to the question.
\end{exa}
\noindent

For representing each agent's awareness and knowledge in the example, we denote by $E,I,$ and $A$ the operators for explicit knowledge, implicit knowledge, and awareness, respectively. For a proposition $\varphi$, write $E_i\varphi$ for ``agent $i$ knows $\varphi$ explicitly.'' A subscript of operators is omitted when the confusion does not occur, and $E_i$, for instance, is denoted simply by $E$. Following the definitions in FH logic, we define $E_i\varphi$ as $I_i\varphi\wedge A_i\varphi$ and $I_i\varphi$ as $K_i\varphi$ in $\mathrm{EL}$. 

Note that, although many propositions appear in this example, we pick up the last inference (``$\ddag$'') that derives $p_4$, describing how implicit knowledge, explicit knowledge, and awareness are associated with the relevant propositions. Furthermore, let us emphasize that the focus of this discussion is on the agents' knowledge and awareness after their reasoning process is complete. We do not address dynamics such as becoming aware of new propositions through reasoning.

Agent $a$, who can construct the proof in Figure 2 in her mind, is aware of all the propositions; $b$, who stops the inference at ``$\dag$'', is unaware of $p_3$. Accordingly, each agent's awareness is expressed as follows: 
\begin{itemize}[leftmargin=17.5mm]
  \item[(1)] $a$ is aware of $p_2,p_3,f_3,p_4$ ($A_a p_2, A_a p_3, A_a f_3, A_a p_4)$; 
  \item[(2)] $b$ is aware of $p_2,f_3,p_4$ ($A_b p_2, \neg A_b p_3, A_b f_3, A_b p_4)$.
\end{itemize}

As for implicit knowledge, the agents have the following knowledge of the basics of elementary geometry from our assumption:
\begin{itemize}[leftmargin=17.5mm]
  \item[(3)] $a$ and $b$ know $f_3$ implicitly ($I_a f_3$, $I_b f_3$);
  \item[(4)] $a$ and $b$ implicitly know that if $p_2$, $p_3$, and $f_3$ hold, then $p_4$ holds, too ($I_a ((p_2\wedge p_3\wedge f_3)\to p_4)$, $I_b ((p_2\wedge p_3 \wedge f_3)\to p_4)$).
\end{itemize}

Since both agents see the diagram and recognize each side and angle, each $k_i$ belongs to their implicit knowledge. Hence, they also implicitly know $p_2$ and $p_3$. For it relies, eventually, solely on three assumptions: (i) Each $f_i$ and $k_i$ belongs to their implicit knowledge; (ii) the implication involved at each inference step itself belongs to their implicit knowledge; (iii) implicit knowledge is closed under implication.
\begin{itemize}[leftmargin=17.5mm]
  \item[(5)] $a$ and $b$ implicitly know both $p_2$ and $p_3$ ($I_a p_2, I_a p_3$, $I_b p_2, I_b p_3$).
\end{itemize}
\noindent

Furthermore, since implicit knowledge is closed under implication, 
\begin{itemize}[leftmargin=17.5mm]
  \item[(6)] $a$ and $b$ know $p_4$ implicitly ($I_a p_4$, $I_b p_4$).
\end{itemize}
\noindent

(1), (2), and (6) derive $E_a p_4$ and $E_b p_4$ from the definition of explicit knowledge.
\begin{itemize}[leftmargin=17.5mm]
  \item[(7)] $a$ knows $p_4$ explicitly ($E_a p_4$).
  \item[(8)] $b$ knows $p_4$ explicitly ($E_b p_4$).
\end{itemize}

While the explicit knowledge (7) for $a$ is consistent with the actual one, the knowledge (8) for $b$ contradicts her available knowledge. The implicit knowledge of $p_4$ is justified through an inference from $p_2$, $p_3$, and $f_3$. Since $b$ is unaware of $p_3$, this justification is not available within her awareness. Agent $b$ thereby excludes $p_4$ from her explicit (available) knowledge. In other words, $b$, who is unaware of the required proposition $p_3$ to solve the problem, does not explicitly know that $p_4$ is true ($\neg E_b p_4$).

This resulting tension between the definition and intuition relates to the valid formula in FH logic, $I(\varphi) \wedge I(\varphi \to \psi) \to (A\psi \to E\psi)$. The formula is incoherent with the view that unless an agent is aware of all the propositions necessary to infer the consequence, it cannot be considered that she knows the consequence explicitly. For instance, for a system, we may take the deductive closure of its knowledge base as implicit knowledge and the information currently in the system's focus as awareness. Its explicit knowledge should be what is obtained through inference using propositions within the current focus, rather than what belongs merely to the deductively closed knowledge set currently focused on. This observation suggests that, in some cases, explicit knowledge cannot be captured merely by filtering unaware propositions out of implicit knowledge. 

\subsection{Awareness and Indistinguishability}
The semantics of $\mathrm{EL}$ is presented by a Kripke model which consists of a non-empty set of possible worlds where each world has a different valuation for propositions, to which each agent may or may not be accessible. The knowledge is defined as a true proposition in all of an agent's accessible worlds. This accessibility is understood to express epistemic indistinguishability among possible worlds. For instance, an agent does not know that $\varphi$ is true ($\neg K_i \varphi$) if and only if $\varphi$ is true not in all her epistemically indistinguishable worlds (epistemic alternatives), which means that she cannot distinguish between the world where $\varphi$ is true and the world where not.
 
The notion of awareness induces its own form of indistinguishability among possible worlds, as demonstrated in the following figure. Let $p$ and $q$ be propositions. Figure 3 illustrates a possible world model where circles express possible worlds, propositions listed above each world true propositions in the world, and a proposition within curly brackets an agent's aware proposition.
 \begin{figure}
   \centering
   % \hspace{-2zh}
   \includegraphics[width=4.5cm]{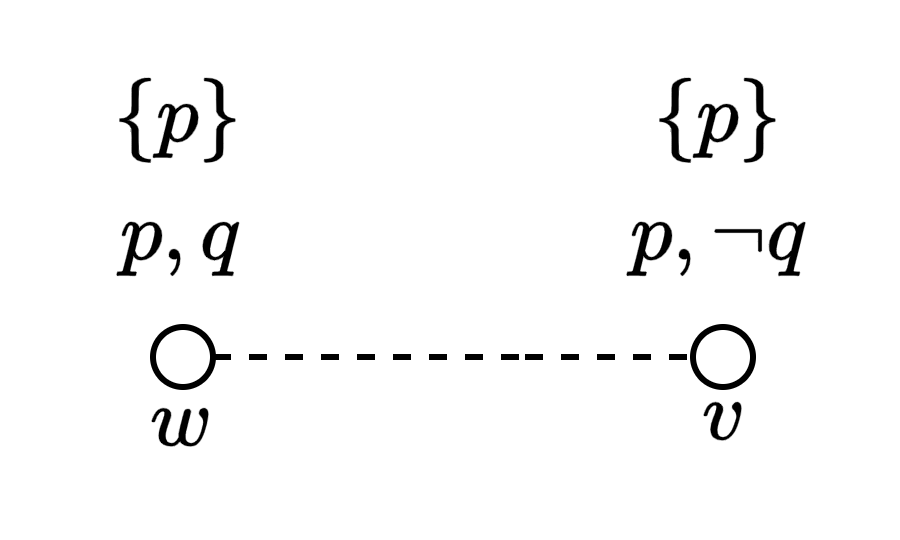}
   \caption{A possible world model where a set of possible worlds is $\{w,v\}$.}
 \end{figure}
 \noindent
The only difference between $w$ and $v$ is the truth value of $q$, and the agent is unaware of $q$ in both worlds. In this situation, she views $w$ and $v$ as the same and does not distinguish between them due to a lack of awareness, which is represented by the dotted line. We call this kind of indistinguishability specifically \textit{awareness-indistinguishability} and show that it plays an important role in the formalization of explicit knowledge.

We consider which possible worlds for each agent are epistemically indistinguishable in the motivating example. Figure 4 illustrates a possible world model for Example 4, depicting three possible worlds and indistinguishability among them. For readability, we only include the four propositions $p_2,p_3, f_3$, and $p_4$ involved in the last inference (``$\ddag$''). The actual world is $w$, indicated by the double circle. Each curly bracket with a subscript of an agent indicates the agent's aware propositions.
\begin{figure}
   \centering
   \includegraphics[width=10cm]{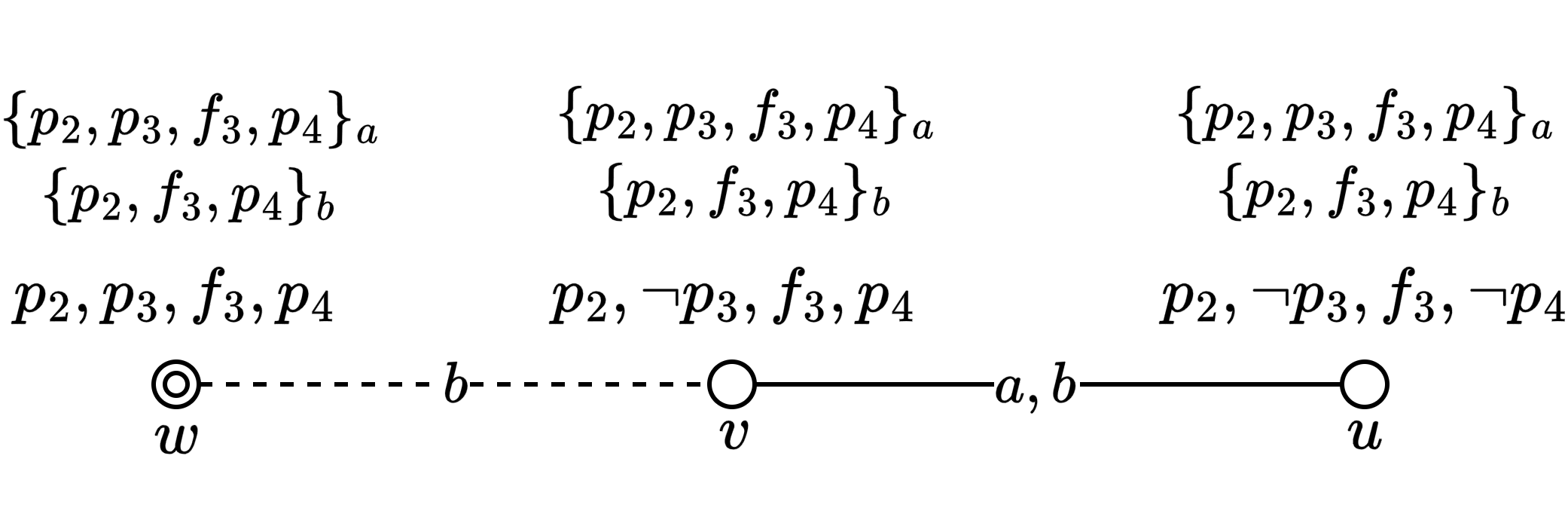}
   \caption{A possible world model where a set of possible worlds is $\{w,v,u\}$.}
 \end{figure}
 \noindent
Since $b$ is unaware only of $p_3$, she does not distinguish between $w$ and $v$. This corresponds to awareness-indistinguishability, attributed to her lack of awareness, represented by the dotted line labeled with $b$ in the figure. Furthermore, in $v$ (where $p_3$ is false), neither agent knows $p_4$ even implicitly because of the absence of ``$\ddag$'', being unable to distinguish the world from $u$ (where $p_4$ is false). This kind of indistinguishability, represented by the solid line labeled with $a$ and $b$ in the figure, expresses epistemic alternatives independent of the lack of awareness (in fact, these are still indistinguishable even under $a$'s complete awareness). Note that since each agent has no way of knowing the other's awareness, we may consider the worlds differing only in the other's awareness from each of $w, v$, and $u$. However, we omit these, as, considering the presently discussed knowledge, they do not affect the description of agents' epistemic situations.

In the figure, the two kinds of indistinguishability jointly express that $b$ cannot distinguish which world is the actual world among worlds, $w$, $v$ (both in which $x = 30\tcdegree$), and $u$ (in which not), and, therefore, epistemic indistinguishability of $b$ with incomplete awareness. The truth of knowledge is evaluated over all epistemic alternatives from the agent's perspective. The observation of Example 4 suggests that the epistemic alternatives for explicit knowledge should be presented through epistemic indistinguishability incorporating awareness-indistinguishability. Therefore, revisiting the conjunctive account of explicit knowledge ($E\varphi\leftrightarrow I\varphi\wedge A\varphi$), in the next section, we provide a formal definition of explicit knowledge through the composed epistemic indistinguishability and show that the formalization refines that in FH logic. 

\section{Awareness-Based Indistinguishability Logic}
In this section, we formally define the syntax and semantics of our language, introducing our logic named \textit{Awareness-Based Indistinguishability Logic} ($\mathcal{AIL}$). Note that when written in calligraphic letters, $\mathcal{AIL}$ refers to the set of formulae that are valid at every model. After showing the properties of the operators, we demonstrate that the formalization faithfully captures the knowledge in the example introduced in the previous section.
\subsection{Language}
\begin{df}
Let $\mathcal{P}$ be a countable set of atomic propositions and $\mathcal{G}$ a finite set of agents. The language $\mathcal{L}_{AIL}(\mathcal{P,G})$ is the set of formulae generated by the following grammar: 
  \begin{align*}
    &\mathcal{L}_{AIL}(\mathcal{P,G}) \ni\varphi::=  p \mid \neg\varphi \mid \varphi\wedge\varphi\mid A_i \varphi \mid I_i\varphi \mid E_i \varphi \mid [\mathop{\approx}]_i\varphi \mid [\circ^+]_i\varphi,
  \end{align*}
  where $p \in \mathcal{P}$ and $i\in\mathcal{G}$. Other logical connectives $\vee$, $\to$, and $\leftrightarrow$ are defined in the usual manner.
\end{df} 

The operators $A$, $I$, and $E$ represent awareness, implicit knowledge, and explicit knowledge operators, respectively.
\begin{itemize}
 \item[$\bullet$] $A_i\varphi$ reads ``agent $i$ is aware of $\varphi$.''
 \item[$\bullet$] $I_i\varphi$ reads ``agent $i$ knows $\varphi$ implicitly.''
 \item[$\bullet$] $E_i\varphi$ reads ``agent $i$ knows $\varphi$ explicitly.'' 
\end{itemize}
\noindent
The other operators $[\mathop{\approx}]_i$ and $[\circ^+]_i$ are technical background ones. These are used to define explicit knowledge and play an important role in the proof of the completeness theorem in Section 5. Their detailed interpretations are explained in the following subsection. 

\subsection{Semantics}
The semantics is presented in a Kripke style, primarily influenced by an awareness structure \cite{fagin1988belief} and the equivalence relation proposed in \cite{yudai2022-1}.
\begin{df}
  An \textit{epistemic model with awareness} $M$ is a tuple $\langle W, \{\mathop{\sim_i}, \mathscr{A}_i\}_{i\in\mathcal{G}}, V\rangle$, where:
     \begin{itemize}
      \item[$\bullet$] $W \text{ is a non-empty set of possible worlds}$;
      \item[$\bullet$] $\mathop{\sim_i}
      \text{ is an equivalence relation on }W$;
      \item[$\bullet$] $\mathscr{A}_i: W \to 2^{\mathcal{P}}$ is an awareness function satisfying that if $(w,v)\in\mathop{\sim_i}$, then $\mathscr{A}_i(w) = \mathscr{A}_i(v)$;
      \item[$\bullet$] $V : \mathcal{P} \to 2^{W}$ is a valuation.
     \end{itemize}
   \end{df}
\noindent
This model consists of a Kripke model in $\mathcal{S}5$ and awareness functions\footnote{This paper adopts $\mathcal{S}5$ (defining accessibility as an equivalence relation) as the foundation of models, for its tractability and widespread use in applications \cite{fagin1995reasoning}. At the same time, $\mathcal{S}5$ has faced criticism for its introspection for knowledge \cite{williamson1992inexact}. However, the primary focus of this paper, as well as that of awareness logic in general, is to tackle the problem of logical omniscience, rather than the idealization of information processing associated with introspection. Using $\mathcal{S}5$ as the foundation allows us to isolate awareness from errors in information processing induced by lack of introspection, enabling a transparent analysis that highlights the role of awareness. Further details for introspection appear in Subsection 3.4.}. In $\mathrm{EL}$, the equivalence relations are simply called accessibility relations and define implicit knowledge. In our logic, we call it \textit{IK-accessibility relation} to avoid confusion because some types of `accessibilities' appear for defining explicit knowledge. The awareness function maps a world into a subset of atomic propositions which an agent is aware of at the world, and this subset $\mathcal{A}_i(w)$ is called an \textit{awareness set} for the agent $i$. The restriction that awareness sets are equal across accessible worlds on $\sim_i$ ensures that agent $i$ always knows her own awareness (awareness introspection), that is, $A_i\varphi\to I_i A_i\varphi$. We refer to this restriction as \textbf{ka} (knowledge implies awareness), following \cite{van2015handbook}.

Moreover, we define another binary relation on $W$. The relation is associated with each agent's awareness set, which is specified by an awareness function, and, therefore, is model-dependent.
\begin{df}
  For each $i\in\mathcal{G}$, \textit{A-equivalence relation} $\mathop{\approx_i}$ on $W$ is defined by $(w,v) \in \mathop{\approx_i} \textit{\  iff,\ }$ $\mathscr{A}_i(w) = \mathscr{A}_i(v)$ and $w \in V(p) \textit{\ iff \ } v \in V(p) \text{ for every } p \in \mathscr{A}_i(w)$.
\end{df}
\noindent
This relation groups possible worlds with different valuations only for unaware propositions and expresses awareness-indistinguishability. The relation is also an equivalence relation, and the equivalence classes collect worlds regarded as `the same' from the agent's viewpoint at a given moment (Figure 5), in other words, as the same states for her current considering propositions. In model $M$ in the figure, since proposition $q$ is not in $i$'s vocabulary, the subjective description of the agent's epistemic situation should be $M'$ in which each equivalence class in $M$ is viewed and `simplified' as one possible world. We call an equivalence class by $\mathop{\approx_i}$, which is an element of the domain of the quotient model of $M$ by $\mathop{\approx_i}$, a \textit{simplified world for $i$}. Denoting an equivalence class of $w$ by relation $R_i$ by $[w]_{R_i}$, we can express a simplified world with respect to $w$ for $i$ as $[w]_{\mathop{\approx_i}}$.
The restriction that awareness sets are equal across reachable worlds on $\approx_i$ ensures a property about $i$'s knowledge of her own awareness, as well as that on $\sim_i$ does.

\begin{figure}[t]
  \centering
  \includegraphics[width=\textwidth]{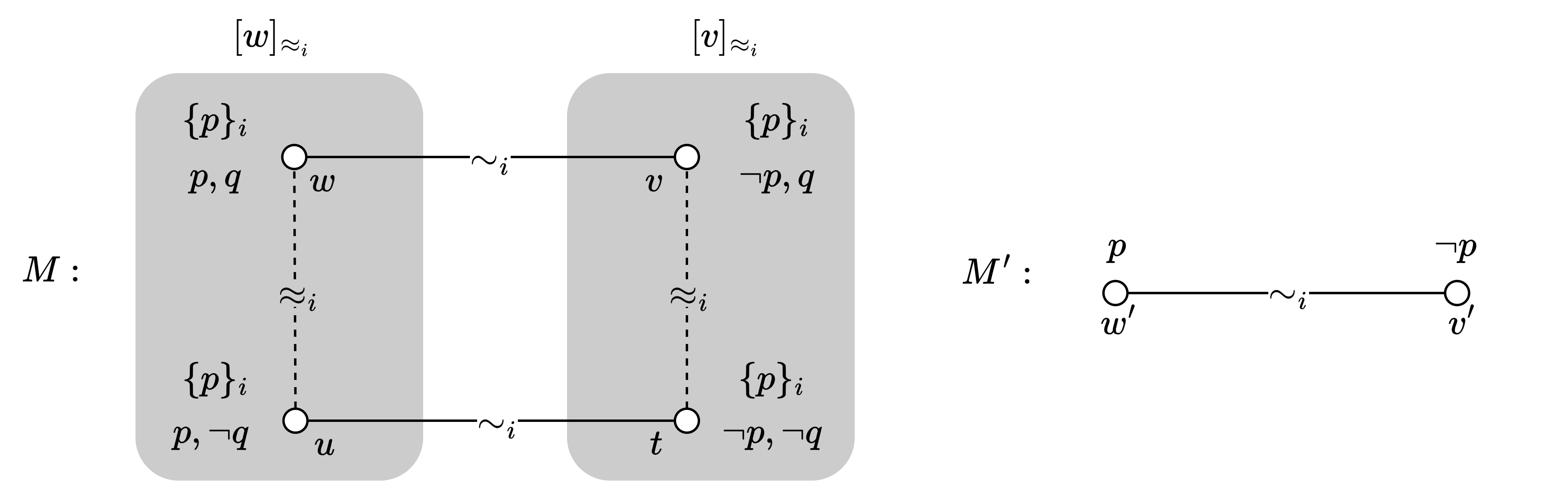}
  \caption{Two possible world models. Epistemic model with awareness $M$ depicts equivalence classes by $\mathop{\approx_i}$. Agent $i$ views $w$ and $u$, and $v$ and $t$ as having no difference due to her lack of awareness of $q$. A Kripke model $M'$ is the model only for $i$'s aware proposition $p$. For readability, we omit the reflexivity.}
\end{figure}

We express epistemic indistinguishability for explicit knowledge by composing IK-accessibility (indistinguishability independent of awareness) with awareness-indistinguishability (indistinguishability dependent on awareness). This resulting relation is called \textit{EK-accessibility relation} and is achieved by taking the transitive closure of the composition of the two relations, $\sim_i$ and $\approx_i$. Formally, the relation is defined as $(\mathop{\sim_i} \circ\mathop{\approx_{i}})^+$, where $\mathop{\sim}_i \circ\mathop{\approx}_i$ the sequential composition of $\mathop{\approx}_i$ and $\mathop{\sim}_i$, for a binary relation $R$, $R^+$ denotes the transitive closure of $R$, and the closure $R^+$ is the smallest set such that $R\subseteq R^+$ and for all $x, y, z$, if $(x,y)\in R^+$ and $(y,z)\in R^+$, then $(x,z)\in R^+$. As shown in Lemma 1, this closure is also an equivalence relation on $W$ and partitions $W$ more coarsely (Figure 6).

\begin{figure}
  \centering
  \includegraphics[width=14cm]{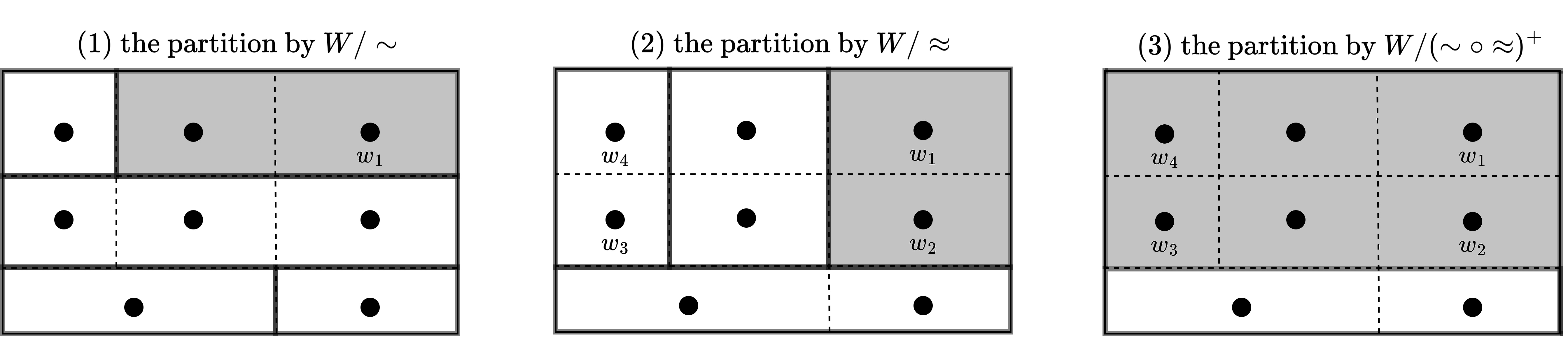}
  \caption{Partitions of a set $W$ consisting of eight worlds by IK-accessibility, A-equivalence, and EK-accessibility relations. Given $[w_1]_{\sim_i}$ and $[w_1]_{\approx_i}$ (see (1) and (2)), $w_2$ in $[w_1]_{\approx_i}$ belongs to the partition by $(\sim_i \circ\approx_i)^{+}$. Then, $w_3$ is in this connected partition as $w_3\sim_i w_2$. Transitively, $w_4$ is also in it as $w_3 \approx_i w_4$. The final partition becomes as (3).
  }
\end{figure}

\begin{lemma}
  Let $\mathop{\sim_i}$ be IK-accessibility relation and $\mathop{\approx_{i}}$ A-equivalence relation. EK-accessibility relation $(\mathop{\sim_i} \circ\mathop{\approx_{i}})^+$ is an equivalence relation.
\end{lemma}
\begin{proof}
  For the reflexivity and the transitivity, the proof is straightforward because both $\mathop{\sim}_i$ and $\mathop{\approx}_i$ are equivalence relations and $(\mathop{\sim}_i\circ\mathop{\approx}_i)^+$ is closed by transitivity. For the symmetricity, It is sufficient to prove, for every $w,v\in W$, if $(w,v)\in(\mathop{\sim}_i\circ\mathop{\approx}_i)^+$, then $(v,w)\in(\mathop{\sim}_i\circ\mathop{\approx}_i)^+$. It follows from the assumption that there exists $u_1,\dots,u_n\in W$ such that $(w,u_1)\in\mathop{\sim}_i\circ\mathop{\approx}_i$, $(u_1,u_2)\in\mathop{\sim}_i\circ\mathop{\approx}_i$, $\cdots$, $(u_{n-1},u_n)\in\mathop{\sim}_i\circ\mathop{\approx}_i$, and $(u_n,v)\in\mathop{\sim}_i\circ\mathop{\approx}_i$. In addtion, for each pair, we take $(w,u_1)$ as an example, there exists $t\in W$ such that $(w,t)\in \mathop{\approx}_i$ and $(t,u_1)\in \mathop{\sim}_i$. Since both relations are equivalence relations, $(u_1,u_1),(t,w)\in \mathop{\approx}_i$, $(u_1,t),(w,w)\in \mathop{\sim}_i$, and $(u_1,t),(t,w)\in\mathop{\sim}_i\circ\mathop{\approx}_i$. Thus, there exists $t_i\in W$, where $1\leq i\leq n+1$, such that $(u_1,t_1),(t_1,w)\in\mathop{\sim}_i\circ\mathop{\approx}_i$, $(u_2,t_2),(t_2,u_1)\in\mathop{\sim}_i\circ\mathop{\approx}_i$, $\cdots$, $(u_n,t_n),(t_n,u_{n-1})\in\mathop{\sim}_i\circ\mathop{\approx}_i$, and $(v,t_{n+1}),(t_{n+1},u_n)\in\mathop{\sim}_i\circ\mathop{\approx}_i$. Therefore, by the transitivity, $(v,w)\in(\mathop{\sim}_i\circ\mathop{\approx}_i)^+$.
\end{proof}

However, in mathematical authenticity, these two relations, $\mathop{\sim}_i$ and $\mathop{\approx}_i$, belong to different concepts. In order to avoid a misleading formalization, we may regard that (i) $\mathop{\approx}_i$ is an $\epsilon$-closure of automata, where we can freely move to different states (worlds) without any computational procedure \cite{sipser2006}. Otherwise, (ii) we may regard that the set of possible worlds $W$ is re-modelled to its quotient space $W/\mathop{\approx}_i$ and after that we may introduce the equivalence $\mathop{\sim}_i$. Nevertheless, because awareness-indistinguishability is static and does not involve any update or re-modelling in epistemic states of agents, we adopt the present formalization.

The satisfaction relation is defined by a pointed model, a pair $(M,w)$ of a model and a world. Let $At(\varphi)$ denote the set of atomic propositions occurring in $\varphi$. 
\begin{df}
  For each epistemic model with awareness $M$ and possible world $w \in W$, the satisfaction relation $\vDash_{AIL}$ is given as follows: 
  \begin{align*}
    M,w \vDash_{AIL} p &\textit{\  iff  \ } w \in V(p) ;\\[-3pt] 
    M,w \vDash_{AIL} \neg \varphi &\textit{\  iff  \ } M,w \nvDash_{AIL}\varphi;\\[-3pt]
    M,w \vDash_{AIL} \varphi\wedge\psi &\textit{\  iff  \ } M,w\vDash_{AIL}\varphi \text{, and } M,w\vDash_{AIL}\psi ; \\[-3pt] 
    M,w \vDash_{AIL} A_{i} \varphi &\textit{\  iff  \ } At(\varphi) \subseteq \mathscr{A}_{i}(w);\\[-3pt]
    M,w \vDash_{AIL} I_i\varphi &\textit{\  iff  \ } M,v\vDash_{AIL} \varphi \text{ for all } v \text{ such that }(w,v)\in \mathop{\sim_i};\\[-3pt]
    M,w\vDash_{AIL} [\approx]_{i}\varphi &\textit{\  iff  \ } M,v\vDash_{AIL}\varphi \text{ for all } v \text{ such that }(w,v)\in \mathop{\approx_{i}};\\[-3pt] 
    M,w\vDash_{AIL} [\circ^+]_{i}\varphi &\textit{\  iff  \ } M,v\vDash_{AIL} \varphi \text{ for all } v\text{ such that } (w,v) \in (\mathop{\sim_i} \circ\mathop{\approx_{i}})^+;\\[-3pt]
    M,w \vDash_{AIL} E_{i} \varphi &\textit{\  iff  \ } M,w \vDash_{AIL} A_{i}\varphi  \text{ and }M,w\vDash_{AIL} [\circ^+]_{i}\varphi.
  \end{align*}
\end{df}
\noindent
When the confusion does not occur, $\vDash_{AIL}$ is denoted simply by $\vDash$.

Each modal operator corresponds to a distinct structural component of the model. In what follows, we summarize their formal definitions and provide intuitive explanations.
\begin{itemize}
  \item By the definition, the awareness of $\varphi$ ($A_i \varphi$) is reduced to awareness for the atomic propositions occurring in $\varphi$ and then interpreted. For instance, ``an agent is aware of $\varphi \wedge \psi$'' if and only if ``she is aware of $\varphi$ and $\psi$.'' This interpretation with the restriction called \textbf{gpp} (awareness generated by primitive propositions) \cite{van2015handbook} provides a representation of awareness in the sense of \textit{awareness of}.
  \item The $[\mathop{\approx}]_i$ operator refers to a true proposition in all $i$'s awareness-indistinguish- able worlds or a simplified world for $i$. $[\mathop{\approx}]_i\varphi$ is true at world $w$ if and only if $\varphi$ is true in all worlds in which each aware proposition has the same truth value as in $w$, unaware propositions disregarded. A formula with this modality added can be understood as the formula with the modality removed in a simplified world. For, $[\mathop{\approx}]_i\varphi$ holds at $w$ precisely when $\varphi$ is true in the simplified world with respect to $w$ for $i$. In other words, this operator refers to equivalence classes representing an `abstraction' of the given model \cite{udatsu}. A deeper analysis of the operator lies beyond our present scope and is left for future work.
  \item The $[\circ^+]_i$ operator refers to a true proposition in all $i$'s indistinguishable worlds in terms of those two kinds of indistinguishability, that is, on the EK-accessibility relation. Using the composition instead of the union makes $[\circ^+]_i$ be equivalent to infinite iteration of $I_i [\mathop{\approx}]_i$. Although a formula with infinite iteration is not well-formed, this observation serves to prove the completeness theorem. 
  \item Implicit knowledge is interpreted similarly to $K$ of $\mathrm{EL}$, through IK-accessibility relation. On the other hand, explicit knowledge is interpreted by an agent's awareness and her epistemically indistinguishable worlds referred to by $[\circ^+]_i$ operator, not by $I_i$ operator.
\end{itemize}
 
Our logic does not validate $I_i\varphi\wedge A_i\varphi \to E_i\varphi$, unlike FH logic. For example, consider model $M$ in Figure 5. Propositions $A_i p$, $I_i p$, and $[\mathop{\approx}]_i p$ hold at $w$ but not $[\circ^+]_i p$. Therefore, $E_i p$ does not hold at $w$. Table 1 outlines the differences in its formalizations. We will discuss some properties of our operators in Subsection 3.4. 

\begin{table}
  \caption{The outline of the differences in formalizations between an FH logic and our logic. The double-lined arrow $:\Longleftrightarrow$ with a colon indicates ``is defined by.'' Note that for readability, some details of the conditions, such as a model and world, are omitted.}	
  \centering
\begin{tabular}{|c|c|c|}
  \hline
 & FH Logic\footnotemark[6] & Our Proposal\\\hline
 Awareness & \multicolumn{2}{|c|}{$A_i\varphi$ $:\Longleftrightarrow$ $At(\varphi)\in \mathscr{A}_i(w)$}\\
 Implicit Knowledge & \multicolumn{2}{|c|}{$I_i\varphi$ $:\Longleftrightarrow$ $w \sim_i v$}\\\hline
   A-equivalence && $[\approx]_i\varphi$ $:\Longleftrightarrow$ $w \approx_i v$\\ 
   EK-accessibility &&$[\circ^+]_i\varphi$ $:\Longleftrightarrow$ $w (\mathop{\sim_i} \circ\mathop{\approx_{i}})^+ v$\\\hline
 Explicit Knowledge & $E_i\varphi:\Longleftrightarrow A_i\varphi \wedge I_i\varphi$ & $E_i\varphi:\Longleftrightarrow A_i\varphi \wedge [\circ^+]_i\varphi$\\\hline 
\end{tabular}
\end{table}

The validity is defined as usual.
\begin{df}
  A formula $\varphi$ is valid at $M$ if $\varphi$ holds at every pointed world in $M$, which is denoted by $M\vDash\varphi$. A formula $\varphi$ is valid if $\varphi$ holds at every pointed model, which is denoted by $\vDash \varphi$. 
\end{df}

\subsection{An Explanation with Dynamic Operators}
\footnotetext[6]{A FH logic with the restriction known as \textbf{gpp} \cite{van2015handbook}. This restriction reduces awareness of a formula to that of the atomic propositions occurring therein and corresponds to the notion of awareness that this paper focuses on, as mentioned in Section 1. Further details appear in Section 4.} 

We discuss more deeply what the information represented by $[\circ^+]$ refers to. The definition clearly indicates that this is stronger than implicit knowledge. However, when an agent is aware of all propositions, both are equivalent to explicit knowledge. To demonstrate the difference between them, we temporarily introduce two simple dynamic operators: a becoming-aware operator and a becoming-unaware operator. These operators have no other functionality besides adding propositions to, or removing them from, an agent's awareness set. Formally, let $\mathcal{Q}\subseteq\mathcal{P}$, and the updated models $M[+\mathcal{Q}]_i$ and $M[-\mathcal{Q}]_i$ are defined as follows: 
  $M[+\mathcal{Q}]_i$ is an epistemic modal with awareness 
  $\langle W, \{\sim_i, \mathscr{A}_i'\}_{i\in\mathcal{G}}, V\rangle$, where, 
  \begin{itemize}
   \item for all $w$, $\mathscr{A}_i'(w) \coloneqq \mathscr{A}_i(w)\cup \mathcal{Q}$;
  \end{itemize}
  $M[-\mathcal{Q}]_i$ is an epistemic modal with awareness 
  $\langle W, \{\sim_i, \mathscr{A}_i'\}_{i\in\mathcal{G}}, V\rangle$, where,
  \begin{itemize}
   \item for all $w$, $\mathscr{A}_i'(w) \coloneqq \mathscr{A}_i(w)\setminus \mathcal{Q}$.
  \end{itemize}
  \noindent 
  The satisfaction relation for each operator is as follows:
  \begin{align*}
	M,w &\vDash [+\mathcal{Q}]_i\varphi
	  \textit{\  iff  \ } M[+\mathcal{Q}]_i,w\vDash \varphi;\\[-2pt]
	M,w &\vDash [-\mathcal{Q}]_i\varphi
    \textit{\  iff  \ } M[-\mathcal{Q}]_i,w\vDash \varphi.
  \end{align*} 
\noindent 

Using these operators, we find that $[\circ^+]\varphi$ implies $[+At(\varphi)]E[-At(\varphi)]\varphi$ ($[\circ^+]\varphi \to [+At(\varphi)]E[-At(\varphi)]\varphi$). That is, the information referred to by $[\circ^+]$ would become explicit knowledge immediately upon becoming aware of that particular proposition. On the other hand, implicit knowledge requires awareness of all propositions. Formally, $I\varphi \leftrightarrow [+\mathcal{P}]E([-\mathcal{P}]\varphi)$. For the implicit knowledge to become explicit, awareness of that proposition alone is insufficient; awareness of all propositions is required. These formulae highlight the distinctions between $I$ and $[\circ^+]$ operators. 

Furthermore, this observation is coherent with the interpretation of implicit knowledge as the ideal knowledge of a logically omniscient agent and corroborates our definition of explicit knowledge. For, the possession of knowledge encompassing all tautologies, along with knowledge derived from the agent's existing knowledge through logical inference, that is, implicit knowledge, is appropriate for an agent capable of referring to, being fully aware of, all propositions at a given time. Such an agent can evaluate the necessary truth of any proposition, including tautologies, and perform logical inference on any given propositions. 

\subsection{Basic Properties}
In this subsection, we show some properties of the operators introduced in the previous subsection, mainly $A, I$, and $E$, including their interactions.

\paragraph*{Awareness} 
The following validity is based on the definition of the awareness operator. 
\begin{itemize}
  \item[$\bullet$] $\vDash A_i\varphi \leftrightarrow 
  \bigwedge_{p\in At(\varphi)} A_i p$.
\end{itemize}
\noindent
The formula means that being aware of a proposition $\varphi$ is equivalent to being aware of the atomic propositions occurring in $\varphi$, capturing the notion of \textit{awareness of} as mentioned in Section 1.

The definitions of $\mathop{\sim_i}$ and $\mathop{\approx_i}$ have the conditions on the equality of awareness sets, from which the following expressions hold:
\begin{itemize}
  \item[$\bullet$] $\vDash I_i A_i \varphi$ and $\vDash I_i \neg A_i \varphi$,
  \item[$\bullet$] $\vDash E_i A_i \varphi$ and $\vDash \neg E_i\neg A_i \varphi$.
\end{itemize}
\noindent
These validities show that awareness is always implicitly and explicitly known, while unawareness does not belong to the agent's explicit knowledge.

\paragraph*{Explicit and Implicit knowledge}
Explicit knowledge in FH logic is defined as the aware implicit knowledge. On the other hand, the following expressions hold in our model:
\begin{itemize}
  \item[$\bullet$] $\vDash E_i\varphi \rightarrow I_i \varphi \wedge A_i\varphi$,
  \item[$\bullet$] $\nvDash I_i \varphi \wedge A_i\varphi\rightarrow E_i\varphi$,
  \item[$\bullet$] For all $w\in W$ in $M$, if $\mathcal{P} = \mathscr{A}_i(w)$, then $M\vDash I_i\varphi \leftrightarrow E_i \varphi$,
\end{itemize}
\noindent
In our logic, explicit knowledge implies implicit knowledge and awareness, but the converse does not hold. This invalidity (the second formula) indicates the case, as shown in Example 4, where an agent knows a proposition implicitly and is aware of that proposition itself, yet it is still not explicitly known. This is one of the apparent differences between FH and our logic, showing that our explicit knowledge is a refinement of that in FH logic. Nevertheless, the third expression holds, as with FH logic. This means that if an agent is aware of all propositions, implicit knowledge is equivalent to explicit knowledge. 

Accordingly, we find the following formula, which is valid in FH logic, to be invalid.
\begin{itemize}
  \item $\nvDash I_i\varphi \wedge I_i(\varphi\to \psi) \to (A_i \psi \to E_i \psi)$.
\end{itemize}
This formula itself means that even when an agent lacks awareness of an inference process, as long as she is aware of its consequence, she classifies that consequence as explicit knowledge. Our logic excludes this formula. In the view introduced in Section 1, the agent does not know the consequence explicitly unless she is aware of all the propositions necessary to infer the consequence.

\paragraph*{Introspection}
Both implicit and explicit knowledge in our logic are interpreted with equivalence relations; hence, for $I$ and $E$, the axioms $\mathrm{T}$ and $\mathrm{4}$ hold. 
\begin{itemize}
  \item[$\bullet$] $\vDash I_i \varphi \rightarrow \varphi$ and $\vDash E_i \varphi \rightarrow \varphi$,
  \item[$\bullet$] $\vDash I_i \varphi \to I_i I_i \varphi$ and $\vDash E_i \varphi \to E_i E_i \varphi$ (Positive introspection).
\end{itemize}
\noindent
However, the axiom $\mathrm{5}$ does not hold for the explicit knowledge. The axiom represents negative introspection: an agent knowing what she does not know. This introspection is viewed as unrealistic for human knowledge.
\begin{itemize}
  \item[$\bullet$] $\vDash \neg I_i \varphi \to I_i \neg I_i \varphi$ (Negative introspection),
  \item[$\bullet$] $\nvDash \neg E_i \varphi \to E_i \neg E_i \varphi$,
  \item[$\bullet$] $\vDash \neg E_i\varphi \wedge A_i \varphi \to E_i \neg E_i \varphi$ (Weak negative introspection).
\end{itemize}
\noindent
The third formula, known as weak negative introspection, is valid, stating that an agent explicitly knows what she does not know explicitly only when she is aware of that proposition. Our logic is $\mathcal{S}5$-based because of several benefits (See the footnote in Definition 2); therefore, as seen, the logic has the property of introspection. This shows that while our awareness logic does address one aspect of the concern that $\mathcal{S}5$ is `too strong' for human knowledge by avoiding logical omniscience, which has been our primary focus, it scarcely addresses the aspect concerning introspection \cite{williamson1992inexact}.

From the definitions, $[\circ^+]_{i}\varphi$ is equivalent to the infinite formula $(I_i [\mathop{\approx}]_i)^1 \varphi \wedge \dots$, where $(I_i [\mathop{\approx}]_i)^n$ is iteration of $I_i [\mathop{\approx}]_i$ $n$ times, therefore $E_{i}\varphi$ is equivalent to the infinite formula $A_{i}\varphi \wedge (I_i [\mathop{\approx}]_i)^1 \varphi \wedge \dots$. Although these formulae are not well-formed, the following expressions follow:
\begin{itemize}
  \item[$\bullet$] $\nvDash I_i \varphi \to E_i I_i \varphi$ (Mix4),
  \item[$\bullet$] $\nvDash \neg I_i \varphi \to E_i \neg I_i \varphi$ (Mix5),
  \item[$\bullet$] $\nvDash I_i \varphi \wedge A_i \varphi \to E_i I_i\varphi$ (Mix4A),
  \item[$\bullet$] $\vDash \neg I_i \varphi \wedge A_i \varphi \to E_i \neg I_i \varphi$ (Mix5A).
\end{itemize}
The first and second formulae express that explicit knowledge of (or about) implicit knowledge/ignorance cannot be derived from the implicit information alone. The third formula, which prefixes awareness to the antecedent of the axiom $\mathrm{4}$, is invalid and shows that implicit knowledge and awareness are insufficient for obtaining explicit knowledge again. On the other hand, the fourth formula, which prefixes awareness to the antecedent of $\mathrm{5}$, is valid.

\paragraph*{Logical omniscience}
As with other logics of awareness, our explicit knowledge avoids characterizations of logical omniscience.
\begin{itemize}
  \item[$\bullet$] $\vDash \varphi$ does not imply $\vDash E_i \varphi$ (Knowledge of valid formulae),
  \item[$\bullet$] $M,w\vDash E_i\varphi$ and $\vDash \varphi\to\psi$ do not imply $M,w\vDash E_i\psi$ (Closure under logical implication).
  \item[$\bullet$] $M,w\vDash E_i\varphi$ and $\vDash \varphi\leftrightarrow\psi$ do not imply $M,w\vDash E_i\psi$ (Closure under logical equivalence).
\end{itemize}
\noindent
Furthermore, the expressions with awareness prefixed to the antecedent hold.
\begin{itemize}
  \item[$\bullet$] $\vDash \varphi \text{ and } M,w\vDash A_i\varphi \text{ imply } M,w\vDash E_i \varphi$,
  \item[$\bullet$] $M,w\vDash E_i\varphi \wedge A_i\psi$ and $\vDash \varphi\to\psi$ imply $M,w\vDash E_i\psi$
  \item[$\bullet$] $M,w\vDash E_i\varphi \wedge A_i\psi$ and $\vDash \varphi\leftrightarrow\psi$ imply $M,w\vDash E_i\psi$
\end{itemize}

\paragraph*{Implication relations}
Some operators imply other operators. The following expressions show some implication relations among modalities.
\begin{itemize}
  \item[$\bullet$] $\vDash E_i\varphi \leftrightarrow A_i\varphi \wedge [\circ^+]_i\varphi$,
  \item[$\bullet$] $\vDash [\circ^+]_i \varphi \to [\mathop{\approx}]_iI_i\varphi \wedge  I_i \varphi \wedge [\mathop{\approx}]_i \varphi$,
  \item[$\bullet$] $\vDash I_i \varphi \wedge A_i\varphi \to [\mathop{\approx}]_i \varphi$,
  \item[$\bullet$] $\nvDash [\mathop{\approx}]_i \varphi \to I_i \varphi$ and $\nvDash [\mathop{\approx}]_i \varphi \to A_i\varphi$.
\end{itemize}
The $[\circ^+]$ operator is implied by the $E$ operator and implies not only the $I$ and $[\mathop{\approx}]$ operators but also $[\mathop{\approx}] I$. Thus, it is a necessary condition for $\varphi$ to be included in explicit knowledge that $\varphi$ remains in the knowledge of $i$ as a logically omniscient agent ($I_i \varphi$), even in the simplified world in which propositions of which $i$ is unaware are disregarded. The $[\mathop{\approx}]$ operator is implied by the $I$ and $A$ operators, but the converse is not. Arranging the operators in the order, we have 
\begin{itemize}
  \item $E = [\circ^+]+A > [\mathop{\approx}]I+A > I+A > [\mathop{\approx}]$, 
\end{itemize}
where $X = Y$ abbreviates $X\varphi\leftrightarrow Y\varphi$ and $X > Y$ $X\varphi\to Y\varphi$ for all $\varphi$.

\subsection{A Representation of the Motivating Example} 
We provide a formal representation of the agents' knowledge in Example 4 with our semantics.

We recall Example 4 in Section 2. The last inference in Figure 2 involves $p_2, p_3, f_3$, and $p_4$, each of which reads:
\begin{itemize}
  \item $p_2$: ``The lengths of side $DA$ and $DC$ are equal;''
  \item $p_3$: ``The sum of the measures of angle $DAC$ and $DCA$ is $60$ degrees;''
  \item $f_3$: ``The measures of angles opposite the equal sides are equal;''
  \item $p_4$: ``The measure of angle $x$ is $30$ degrees.''
\end{itemize}

Reflecting their awareness and knowledge, the IK-accessibility relations should be:
\begin{itemize}
 \item $\mathop{\sim_a} = \mathop{\sim_b} = \{(w,w),(v,v),(u,u),(v,u),(u,v)\}$.
\end{itemize} 
\noindent
As for awareness, since $b$ is unaware of $p_2$ at all worlds, 
\begin{itemize}
 \item $\mathscr{A}_a(w) = \mathscr{A}_a(v) = \mathscr{A}_a(u) = \{p_2, p_3, f_3, p_4\}$;
 \item $\mathscr{A}_b(w) = \mathscr{A}_b(v) = \mathscr{A}_b(u) = \{p_2, f_3, p_4\}$;
\end{itemize} 
\noindent
From the definition, the A-equivalence relations should be: 
\begin{itemize}
  \item $\mathop{\approx_a} = \{(w,w),(v,v),(u,u)\}$;
  \item $\mathop{\approx_b} = \{(w,w),(v,v),(u,u),(w,v),(v,w)\}$;
\end{itemize} 

\begin{figure}
  \centering
  \includegraphics[width=9cm]{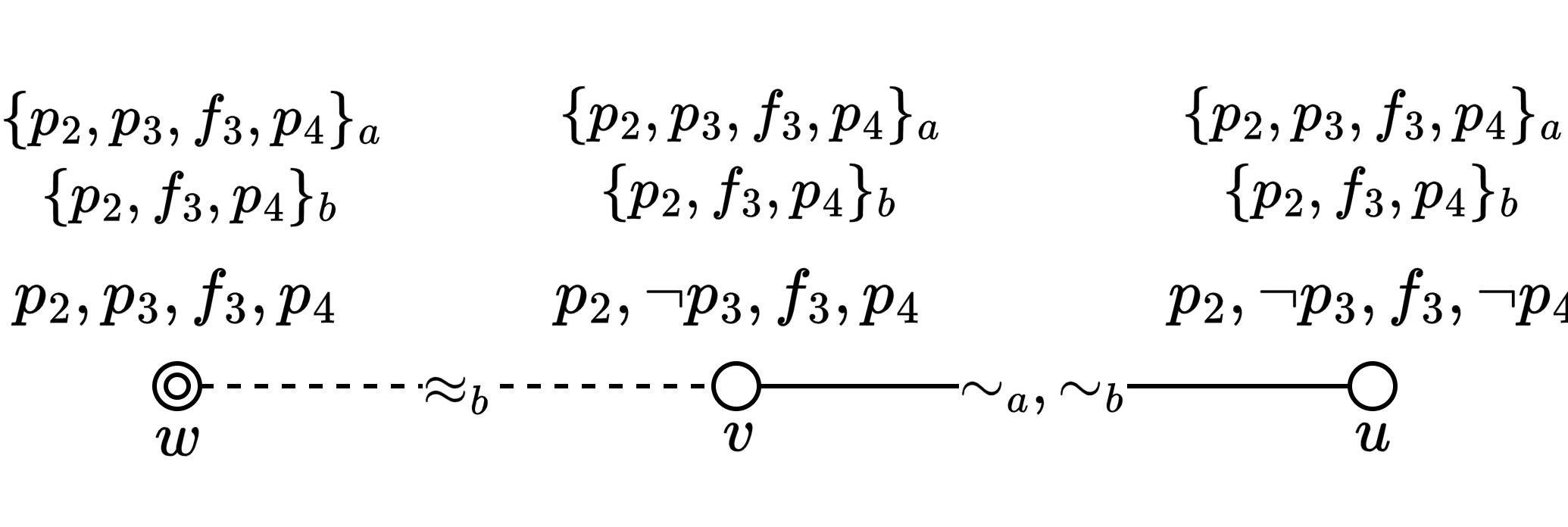}
  \caption{An epistemic model with awareness for Example 4. For readability, we omit the reflexivity.}
\end{figure}
\noindent

These are illustrated in Figure 7. In the actual world $w$, implicit knowledge, awareness, and explicit knowledge of each agent ((1),$\dots$,(7) and $\neg E_b p_4$ instead of (8)) are correctly represented, that is, for $b$, we have $M,w\vDash \neg A_b p_3 \wedge A_b (p_2 \wedge f_3 \wedge p_4) \wedge I_b (p_2 \wedge p_3 \wedge f_3) \wedge I_b(p_2 \wedge p_3 \wedge f_3 \to p_4)\wedge I_b p_4 \wedge \neg E_b p_4$. This is consistent with the intuition about $b$'s explicit knowledge, which is that agent $b$, who is unaware of the required proposition to answer the question, does not explicitly know the answer. As for $a$, we have $M,w\vDash E_a p_4$, and the consistency with her knowledge is maintained.

\begin{table}
  \centering
 \caption{Implicit knowledge, explicit knowledge, and awareness of each agent in Example 4. Note that this table is not exhaustive; for instance, a logical implication involved in ``$\dag$'' and those such as $k_1 \wedge k_2$ are not included.}
 \resizebox{\textwidth}{!}{
  \begin{tabular}{l|l|l}
    & Agent $a$ & Agent $b$ \\\hline\hline
    Implicit knowledge & $f_i$,$k_i$,$p_i$,$p_2\wedge p_3 \wedge f_3 \to p_4$ & $f_i$,$k_i$,$p_i$,$p_2 \wedge p_3 \wedge f_3 \to p_4$\\\hline
    Awareness & $f_i$,$k_i$,$p_i$,$p_2 \wedge p_3 \wedge f_3 \to p_4$ & $f_1$,$f_3$,$k_i$,$p_1$,$p_2$,$p_4$\\\hline
    EK-accessibility ($[\circ^+]$) & $f_i$,$k_i$,$p_i$,$p_2\wedge p_3\wedge f_3 \to p_4$ & $f_1$,$f_3$,$k_i$,$p_1$,$p_2$,$p_2\wedge p_3\wedge f_3 \to p_4$\\\hline
    Explicit knowledge in FH logic& $f_i$,$k_i$,$p_i$,$p_2\wedge p_3\wedge f_3 \to p_4$ & $f_1$,$f_3$,$k_i$,$p_1$,$p_2$,$p_4$\\\hline
    Our explicit knowledge & $f_i$,$k_i$,$p_i$,$p_2\wedge p_3\wedge f_3 \to p_4$ & $f_1$,$f_3$,$k_i$,$p_1$,$p_2$\\\hline
  \end{tabular}
 }
\end{table}

\section{A Comparison with FH logic: Expressivity and Embedding}
This section compares our logic with an FH logic in terms of expressivity and then proves that the FH logic is embeddable into our logic.

FH logic has all the formulae of $\mathcal{L}_{AIL}(\mathcal{P,G})$ excluding $[\mathop{\approx}]$ and $[\circ^+]$ operators. Formally, the language $\mathcal{L}_{FH}(\mathcal{P,G})$ is the set of formulae generated by the following grammar: 
  \begin{align*}
    &\mathcal{L}_{FH}(\mathcal{P,G}) \ni\varphi
    ::=  p \mid \neg\varphi \mid \varphi\wedge\varphi\mid A_i \varphi \mid I_i\varphi \mid E_i \varphi,
  \end{align*}
where $p \in \mathcal{P}$ and $i\in\mathcal{G}$. Other logical connectives $\vee$, $\to$, and $\leftrightarrow$ are defined in the usual manner. When it is unambiguous from the context, we omit $\mathcal{P,G}$ and write $\mathcal{L}_{AIL}$ instead of $\mathcal{L}_{AIL}(\mathcal{P,G})$ and $\mathcal{L}_{FH}$ $\mathcal{L}_{FH}(\mathcal{P,G})$. To avoid confusion, superscripts indicating the languages are attached to operators in formulae, such as $E^{FH}_i$ and $E^{AIL}_i$. 

FH logic is usually interpreted by awareness structures. Here, we focus on the FH logic interpreted by the structure with restrictions regarding awareness and knowledge, \textbf{gpp} and \textbf{ka}, which are focused on in most applications \cite{van2015handbook,belardinelli2024implicit}. Under the restrictions \textbf{gpp} and \textbf{ka}, the resulting awareness structure is structurally identical to an epistemic model with awareness. Therefore, the semantics can also be given by an epistemic model with awareness, as with $\mathcal{AIL}$. The satisfaction relation $\vDash_{FH}$ is defined in the same way as that in $\mathcal{AIL}$, except for the $E$ operator. 
\begin{align*}
  M,w \vDash_{FH} E_{i} \varphi &\textit{\  iff  \ } M,w \vDash_{FH} A_{i}\varphi  \text{ and }M,w\vDash_{FH} I_{i}\varphi,     
\end{align*}

First, we prove that our logical language is at least as expressive as that of the FH logic.
\begin{df}
  Two formulae $\varphi$ and $\psi$ are equivalent, if for every model and world, $M,w\vDash\varphi$ \textit{iff} $M,w\vDash\psi$, which is denoted by $\varphi\equiv \psi$. A language $\mathcal{L}_2$ is at least as expressive as $\mathcal{L}_1$, if for every $\varphi_1\in\mathcal{L}_1$, there is $\varphi_2\in\mathcal{L}_2$ such that $\varphi_1\equiv \varphi_2$, which is denoted by $\mathcal{L}_1 \preceq \mathcal{L}_2$. $\mathcal{L}_1$ is more expressive than $\mathcal{L}_2$, if $\mathcal{L}_1 \preceq \mathcal{L}_2$ and $\mathcal{L}_2 \not\preceq \mathcal{L}_1$, which is denoted by $\mathcal{L}_1\prec\mathcal{L}_2$.
\end{df}

\begin{lemma}
  $\mathcal{L}_{FH}\preceq\mathcal{L}_{AIL}$.
\end{lemma}
\begin{proof}
  It is sufficient to show a translation $t: \mathcal{L}_{FH}\to\mathcal{L}_{AIL}$ such that $t(\varphi)$ is equivalent to $\varphi$ for every $\varphi\in\mathcal{L}_{FH}$. This is proved by induction on the structure of formulae. The only non-trivial case is the $E$ operator. The translation for this operator is provided by $t(E^{FH}_i\varphi) = A^{AIL}_i t(\varphi) \wedge I^{AIL}_i t(\varphi)$. 
\end{proof}
\noindent
Thus, for every formula of the FH logic, there exists a formula of $\mathcal{AIL}$ that has the same meaning as that formula.

To prove having more expressive power, we first introduce bisimulation \cite{van2007dynamic} between our models and show that this relation implies modal equivalence with respect to the FH logic. This bisimulation is a special case of \textit{standard bisimulation} proposed by \cite{van2018implicit} and suffices for our purpose.
\begin{df}
  A bisimulation between epistemic models with awareness $M = \langle W, \allowbreak \{\sim_i, A_i\}_{i\in \mathcal{G}}, V \rangle$ and $M' = \langle W', \{\sim_i', A_i'\}_{i\in\mathcal{G}}, V'\rangle$ is a relation $B$ on $W \times W'$ satisfying that, for every $p\in\mathcal{P}$ and $i\in\mathcal{G}$, if $(w,w') \in B$, then: 
  \begin{itemize}
    \item $w \in V(p)$ iff $w'\in V'(p)$;
    \item if $(w,v)\in\sim_i$, then there exists $v'$ such that $(w',v')\in\sim_i'$ and $(v,v')\in B$;
    \item if $(w',v')\in\sim_i'$, then there exists $v$ such that $(w,v)\in\sim_i$ and $(v,v')\in B$;
    \item $\mathscr{A}_i(w) = \mathscr{A}'_i(w')$.
  \end{itemize}
\end{df}
\noindent
We write $(M,w)\rightleftharpoons(M',w')$, if and only if there exists a bisimulation $B$ between $M$ and $M'$ such that $(w,w')\in B$.

\begin{lemma}
  For all every $\varphi \in \mathcal{L}_{FH}$, if $(M,w)\rightleftharpoons(M',w')$, then $(M,w)\vDash_{FH}\varphi$ iff $(M',w')\vDash_{FH}\varphi$.
\end{lemma}
\begin{proof}
  We prove it by induction on the structure of formulae. As for the base case and that of logical connectives, these proofs are straightforward from the first condition of the bisimulation. 
  \begin{itemize}
    \item For the case of $I_i\varphi$, suppose that $(M,w)\vDash_{FH}I_i\varphi$. For any $v'$ such that $(w',v')\in\sim'_i$, from the third condition, there exists $v$ such that $(w,v)\in\sim_i$ and $(v,v')\in B$. Given that for all $v$ such that $(w,v)\in \sim_i$, $M,v\vDash_{FH}\varphi$, it follows from the induction hypothesis that $M',v'\vDash_{FH} \varphi$. Since $v'$ is an arbitrary world such that $(w',v')\in\sim'_i$, thus, $M',w'\vDash_{FH} I_i\varphi$. The converse direction can be proved similarly by the second condition.
    \item For the case of $A_i\varphi$, suppose that $(M,w)\vDash_{FH}A_i\varphi$, then $At(\varphi)\subseteq \mathscr{A}_i(w)$. From the fourth condition, $\mathscr{A}_i(w) = \mathscr{A}'_i(w')$. Thus, $At(\varphi)\subseteq \mathscr{A}_i(w')$ and $(M',w')\vDash_{FH}A_i\varphi$.
    \item For the case of $E_i\varphi$, It can be proved easily by decomposing into $A_i\varphi$ and $I_i\varphi$.
  \end{itemize}
  \noindent
\end{proof}

\begin{lemma}
  $\mathcal{L}_{AIL}\not\preceq\mathcal{L}_{FH}$.
\end{lemma}
\begin{proof}
  Consider $M$ and $M'$ in Figure 8. Model $M$ consists of three worlds and $M'$ of a single world. Since $(M,w)\leftrightharpoons(M',w')$, for every $\varphi\in \mathcal{L}_{FH}$, $M,w\vDash_{FH}\varphi$ iff $M',w'\vDash_{FH}\varphi$. On the other hand, $M,w\vDash_{AIL} E_i p$ and $M',w'\vDash_{AIL} \neg E_i p$. Therefore, no formula equivalent to $E_i p$ exists.
\end{proof}

\begin{figure}
  \centering
  \includegraphics[width=13cm]{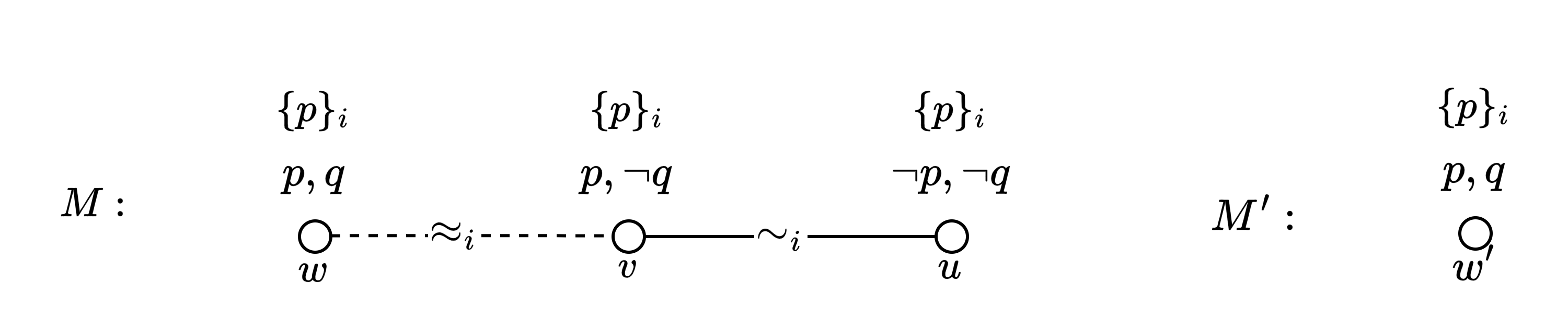}
  \caption{Two epistemic models with awareness for a single agent $i$. Pointed model $(M,w)$ is bisimilar to $(M',w')$. For readability, we omit the reflexivity.}
\end{figure}

\begin{thm}
  $\mathcal{L}_{FH}\prec\mathcal{L}_{AIL}$.
\end{thm}
\begin{proof}
  By Lemma 2 and 4, trivial.
\end{proof}

Furthermore, we prove that the FH logic is embeddable into $\mathcal{AIL}$ by utilizing the translation constructed in Lemma 2.
\begin{thm}
  For every $\varphi\in\mathcal{L}_{FH}$, 
  $\vDash_{FH} \varphi$ iff $\vDash_{\mathcal{AIL}} \varphi$.
\end{thm}
\begin{proof}
  We construct a translation $t: \mathcal{L}_{FH}\to\mathcal{L}_{AIL}$ such that for all models and worlds, $M,w\vDash_{FH}\varphi$ \textit{iff} $M,w\vDash_{\mathcal{AIL}} t(\varphi)$ as well as Lemma 2. This completes the proof. 
\end{proof}

\section{Hilbert System $\mathbf{AIL}$}
We now move on to the proof theory. The operators for awareness, implicit knowledge, and $[\mathop{\approx}]$ can be axiomatized based on the axiomatic system\footnotemark[7] $\mathbf{S5}$ \cite{chellas1980modal} and that of FH logic \cite{halpern2001alternative}. Even though all of the operators apart from $A$ are necessity operators on equivalence relations, the $[\circ^+]$ operator has interactions with other operators and makes the compactness of $\mathcal{AIL}$ lost; therefore, its complete axiomatization cannot be a mere $\mathbf{S5}$ with awareness and multimodal operators.

\footnotetext[7]{When written in bold face, $\mathbf{S5}$ refers to the axiomatic system, distinguished from the class $\mathcal{S}5$ of models.}

Table 3 presents a Hilbert system $\mathbf{AIL}$. The seven axioms about awareness $\mathrm{AN}$, $\mathrm{AC}$, $\mathrm{AA}$, $\mathrm{AI}$, $\mathrm{A[\approx]}$, $\mathrm{A[\circ^+]M}$, and $\mathrm{AE}$ reduce awareness for complex formulae into awareness for atomic propositions forming that formula, corresponding to \textbf{gpp}. The axioms $\mathrm{I[\mathop{\approx}]A}$ and $\mathrm{I[\mathop{\approx}]NA}$ correspond to \textbf{ka}. The axiom $AA[\approx]$ reflects the definition of A-equivalence relation and means that an aware true proposition at world $w$ remains in all worlds regarded as the same worlds as $w$ from the agent's viewpoint. For $\mathrm{K_I, T_I, 5_I, K_{[\approx]}, T_{[\approx]}},$ and $\mathrm{5_{[\approx]}}$, as $I$ and $[\approx]$ are necessity operators on equivalence relations, we adopt $\mathrm{K}$,$\mathrm{T}$, and $\mathrm{5}$ in modal logic for these. The axioms for the $[\circ^+]$ operator, $\mathrm{K_{[\circ^+]}}, \mathrm{MIX}$, and $\mathrm{IND}$, are based on axioms of a logic with common knowledge \cite{van2007dynamic} because both operators are interpreted by reflexive-transitive closures. Together, these lead to the validity of $\mathrm{T}$ and $\mathrm{4}$ for the $E$ operator. The axiom $\mathrm{EA[\circ^+]}$ corresponds to the requirement for the $E$ operator on the satisfaction relation.

\begin{df}
  A system $\mathbf{AIL}$ is a set of formulae containing the axioms in Table 3 and is closed under the inference rules in the table. We write $\vdash\varphi$ if $\varphi\in$ $\mathbf{AIL}$. Let $\Gamma$ be a set of formulae in $\mathcal{L}_{AIL}$ and $\bigwedge\Gamma$ be an abbreviation of $\bigwedge_{\varphi\in \Gamma}\varphi$. If there is a finite subset $\Gamma'$ of $\Gamma$ such that $\vdash\bigwedge\Gamma'\to \varphi$, we write $\Gamma\vdash\varphi$ and call it that $\varphi$ is deducible from $\Gamma$.
\end{df}

\begin{table}[h]
  \caption{Axiom schemata and inference rules of $\mathbf{AIL}$}	
  \centering
   \begin{tabular}{|l|l|}
    \hline
    \multicolumn{2}{|c|}{Axiom schemata}\\\hline
    $\mathrm{TAUT}$ & The set of propositional tautologies\\
        $\mathrm{AN}$ & $A_i\varphi \leftrightarrow A_i\neg\varphi$\\
        $\mathrm{AC}$ & $A_i(\varphi\wedge\psi) \leftrightarrow A_i\varphi \wedge A_i\psi$\\
        $\mathrm{AA}$ & $A_i\varphi \leftrightarrow A_i A_j\varphi$\\
        $\mathrm{AI}$ & $A_i \varphi\leftrightarrow A_i I_j \varphi$\\
        $\mathrm{A[\approx]}$ & $A_i\varphi \leftrightarrow A_i [\approx]_j\varphi$\\
        $\mathrm{A[\circ^+]}$ & $A_i\varphi \leftrightarrow A_i [\circ^+]_j\varphi$\\
        $\mathrm{AE}$ & $A_i \varphi \leftrightarrow A_i E_j \varphi$\\
        $\mathrm{IA}$ & $A_i \varphi\to I_i A_i \varphi$\\        
        $\mathrm{INA}$ & $\neg A_i \varphi\to I_i \neg A_i \varphi$\\
        $\mathrm{AA[\approx]}$ & $A_i p\wedge p\to [\approx]_i p$\\
        $\mathrm{K_I}$ & $I_i(\varphi\to \psi)\to (I_i\varphi \to I_i\psi)$\\
        $\mathrm{T_I}$ & $I_i \varphi \to \varphi$\\
        $\mathrm{5_I}$ & $\neg I_i\varphi\to I_i \neg I_i\varphi$\\
        $\mathrm{K_{[\approx]}}$ & $[\approx]_i(\varphi\to \psi)\to ([\approx]_i\varphi \to [\approx]_i\psi)$\\
        $\mathrm{T_{[\approx]}}$ & $[\approx]_i \varphi \to \varphi$\\
        $\mathrm{5_{[\approx]}}$ & $\neg [\approx]_i\varphi\to [\approx]_i \neg [\approx]_i\varphi$\\
        $\mathrm{K_{[\circ^+]}}$ & $[\circ^+]_i(\varphi\to \psi)\to ([\circ^+]_i\varphi \to [\circ^+]_i\psi)$\\
        $\mathrm{MIX}$ & $[\circ^+]_i\varphi \to \varphi\wedge[\approx]_i I_i [\circ^+]_i\varphi$\\
        $\mathrm{IND}$ & $[\circ^+]_i(\varphi\to [\approx]_i I_i\varphi)\to(\varphi \to [\circ^+]_i\varphi)$\\
        $\mathrm{EA[\circ^+]}$ & $E_i\varphi\leftrightarrow A_i\varphi\wedge [\circ^+]_i\varphi$\\\hline
        \multicolumn{2}{|c|}{Inference Rules}\\\hline
        $\mathrm{MP}$ & If $\vdash \varphi$ and $\vdash \varphi\to\psi$, then $\vdash \psi$\\
        $\mathrm{GI}$ & If $\vdash \varphi$ then $\vdash I_i\varphi$\\
        $\mathrm{G[\approx]}$ &If $\vdash \varphi$ then $\vdash [\approx]_i\varphi$\\
        $\mathrm{G[\circ^+]}$ & If $\vdash \varphi$ then $\vdash [\circ^+]_i\varphi$\\\hline
  \end{tabular}
\end{table}

\subsection{Soundness}
We prove that every theorem of $\mathbf{AIL}$ is valid, which is the soundness of the system.  
\begin{thm}
  If $\vdash\varphi$, then $\vDash\varphi$.
  \end{thm}
  \begin{proof} 
    We prove it by induction on the structure of $\mathbf{AIL}$. First, we prove that all the axioms are valid. For the $I_i$ and $[\approx]_i$ operators, those axioms can be proved in a way similar to \textbf{S5} because both are necessity operators on equivalence relations.
    \begin{itemize}
      \item Since the axioms $\mathrm{AC}$, $\mathrm{AA}$, $\mathrm{AI}$, $\mathrm{A[\approx]}$, $\mathrm{A[\circ^+]}$, and $\mathrm{AE}$ can be proved in a way similar to $\mathrm{AN}$, we show only the proof for $\mathrm{AN}$ here. Suppose that $M,w\vDash A_i\varphi$, then $At(\varphi)\subseteq \mathscr{A}_i(w)$. Since $At(\varphi) = At(\neg\varphi)$, $At(\neg\varphi)\subseteq \mathscr{A}_i(w)$. Thus, $M,w\vDash A_i\neg\varphi$. The converse direction can be proved similarly.
      \item Since the axiom $\mathrm{IA}$ can be proved in a way similar to $\mathrm{INA}$, we show only $\mathrm{INA}$. Suppose that $M,w\vDash \neg A_i \varphi$, then $At(\neg \varphi)\subseteq \mathscr{A}_i(w)$. An awareness set is constant within both equivalence classes defined by $\mathop{\sim_i}$ and $\mathop{\approx_i}$, respectively, from the definitions. For all $v$ such that $(w,v)\in \mathop{\sim_i}$, $At(\neg \varphi)\subseteq \mathscr{A}_i(v)$ and thus $M,w\vDash I_i \neg A_i \varphi$.
      \item For $\mathrm{AA[\approx]}$, suppose that $M,w\vDash A_i p \wedge p$, then $p \in \mathscr{A}_i(w)$ and $w\in V(p)$. For $v$ such that $(w,v)\in\mathop{\approx}_i$, $p\in \mathscr{A}_i(v)$ and $v \in V(p)$ from the definition of the relation. Thus, $M,w\vDash [\approx]_i p$.
      \item For $\mathrm{K_{[\circ^+]}}$, suppose that $M,w\vDash [\circ^+]_i(\varphi\to\psi)$ and $M,w\vDash [\circ^+]_i\varphi$. For all $v$ such that $(w,v)\in (\mathop{\sim_i}\circ\mathop{\approx_i})^+$, $M,v\vDash\varphi\to\psi$ and $M,v\vDash\varphi$, then $M,v\vDash\psi$. Thus, $M,w\vDash [\circ^+]_i\psi$.
      \item For $\mathrm{MIX}$, suppose that $M,w\vDash [\circ^+]_i\varphi$. Since $(\mathop{\sim_i} \circ\mathop{\approx_i})^+$ is an equivalence relation, $M,w\vDash\varphi\wedge[\approx]_i I_i [\circ^+]_i\varphi$. 
      \item For $\mathrm{IND}$, suppose that $M,w\vDash [\circ^+]_i(\varphi\to [\approx]_i I_i\varphi)$ and $M,w\vDash\varphi$. For all $v$ such that $(w,v)\in (\mathop{\sim_i} \circ\mathop{\approx_i})^+$, $M,v\vDash\varphi\to[\approx]_i I_i\varphi$. Since this closure is an equivalence relation, $M,w\vDash[\approx]_i I_i\varphi$. The formula $\varphi$ holds at all worlds from $w$ on $\mathop{\sim_i} \circ\mathop{\approx_i}$, and $[\approx]_i I_i\varphi$ holds at those worlds as well. Therefore, $M,w\vDash [\circ^+]_i\varphi$.
    \end{itemize}
    The remaining task is to prove that if the antecedent is valid, the consequence is also valid for each inference rule. All of them are straightforward.
  \end{proof}  
  
\subsection{Completeness}
We prove the converse direction: every valid formula is a theorem of the system. The logic $\mathcal{AIL}$ is no longer compact because we can take a set of formulae such as $\Phi = \{([\approx]_i I_i)^n \varphi\mid n\in \mathbb{N}\} \cup \{\neg [\circ^+]_i\varphi\}$, where $([\approx]_i I_i)^n$ is $n$ iterations of $[\approx]_i I_i$. Therefore, we restrict elements of maximal consistent sets to finite sets and then define a canonical model. This technique is employed in the proof on a logic with common knowledge, which is defined by the reflexive-transitive closure \cite{van2007dynamic}. We customize the technique for our logic, similar to \cite{yudai2022-1}.

As the first step, we define a \textit{closure} as a restricted set of formulae.
\begin{df}
  Let $cl : \mathcal{L}_{AIL}\to 2^{\mathcal{L}_{AIL}}$ be the function such that, for every $\varphi\in\mathcal{L}_{AIL}$, $cl(\varphi)$ is the smallest set satisfying that: 
  \begin{itemize}
    \item[1.] $\varphi\in cl(\varphi)$; 
    \item[2.] If $\psi\in cl(\varphi)$, then $sub(\psi)\subseteq cl(\varphi)$, where $sub(\psi)$ is the set of subformulae of $\psi$;
    \item[3.] If $\psi\in cl(\varphi)$ and $\psi$ is not a form of negation, then $\neg\psi\in cl(\varphi)$;
    \item[4.] If $A_i \psi \in cl(\varphi)$, then $A_i\chi\in cl(\varphi)$, where $\chi\in sub(\psi)$;
    \item[5.] If $A_i \psi \in cl(\varphi)$, then $I_i A_i \psi, I_i \neg A_i \psi, [\approx]_i p\in cl(\varphi)$, where $p$ is an atomic proposition in $sub(\psi)$;
    \item[6.] If $I_i\psi\in cl(\varphi)$ and $\psi$ is a form of neither $I_i\chi$ nor $\neg I_i\chi$, then $I_i I_i\psi$, $I_i\neg I_i\psi\in cl(\varphi)$;
    \item[7.] If $[\approx]_i\psi\in cl(\varphi)$ and $\psi$ is a form of neither $[\approx]_i\chi$ nor $\neg [\approx]_i\chi$, then $[\approx]_i[\approx]_i\psi$, $[\approx]_i\neg[\approx]_i\psi\in cl(\varphi)$;
    \item[8.] If $[\circ^+]_i\psi\in cl(\varphi)$, then $[\approx]_i I_i [\circ^+]_i\psi\in cl(\varphi)$;
    \item[9.] If $E_i\psi\in cl(\varphi)$, then $A_i \psi$, $[\circ^+]_i \psi\in cl(\varphi)$.
  \end{itemize} 
\end{df}
\noindent

\begin{lemma}
  For every $\varphi$, $cl(\varphi)$ is finite.
\end{lemma}
\begin{proof}
  We prove it by induction on the structure of $cl(\varphi)$. Since the number of subformulae of each formula is finite, all the conditions preserve the finiteness.
\end{proof}

We proceed to define a consistent set and a maximal consistent set in a closure. A set of the latter is used as a domain of a canonical model.
\begin{df}
  Let $\Phi$ be the closure of a formula. 
  $\Gamma$ is a consistent set in $\Phi$ iff 
  \begin{itemize}
    \item[$\bullet$] $\Gamma \subseteq \Phi$, 
    \item[$\bullet$] $\Gamma \nvdash \bot$.
  \end{itemize}
  Moreover, $\Gamma$ is a maximal consistent set in $\Phi$, iff $\Gamma$ is a consistent set and
  \begin{itemize}
    \item[$\bullet$] There is no $\Gamma'\subseteq\Phi$ such that $\Gamma\subset\Gamma'$ and $\Gamma'\nvdash \bot$.
  \end{itemize}
\end{df}

A consistent set can always expand to a maximal consistent set that includes the original set. 
\begin{lemma}
  Let $\Phi$ be the closure of a formula. If $\Gamma$ is a consistent set in $\Phi$, then there exists a maximal consistent set $\Delta$ in $\Phi$ such that $\Gamma\subseteq\Delta$. 
\end{lemma}
\begin{proof}
  It can be proved in the same way as Lindenbaum's lemma. We obtain a maximal consistent set by adding formulae in $\Phi$ to $\Gamma$ so as to preserve the consistency.
\end{proof}

Next, we define a canonical model restricted by the closure of a formula and an equivalence relation depending on the model.
\begin{df}
  Let $\Phi$ be the closure of a formula. A canonical model $M^*$ for $\Phi$ is a tuple $\langle W^*,\{(\mathop{\sim_i})^*,(\mathscr{A}_i)^*\}_{i\in\mathcal{G}}, V^*\rangle$, 
  where:
    \begin{itemize}
      \item $W^*$ $\coloneqq\{\Gamma \mid \Gamma \text{ is a maximal consistent set in }\Phi\}$;
      \item $(\Gamma,\Delta)\in (\mathop{\sim_i})^* \textit{\  iff  \ }\{\varphi\mid I_i\varphi\in \Gamma\}\subseteq \Delta$;
      \item $V^*(p)\coloneqq \{\Gamma \mid p\in\Gamma\}$;
      \item $(\mathscr{A}_i)^*(\Gamma) \coloneqq \{p\mid A_i p\in \Gamma\}$.
    \end{itemize}
\end{df}

\begin{df}
  Let $\Phi$ be the closure of a formula. A relation $(\mathop{\approx_i})^*$ on $W^*$ is defined by $(\Gamma,\Delta) \in (\mathop{\approx_i})^* \textit{\  iff  \ } \{\varphi\mid [\approx]_i\varphi\in \Gamma\}\subseteq \Delta$.
\end{df}

We now prove the canonicity: a canonical model for a closure is an epistemic model with awareness.
\begin{lemma}
  For every $\varphi$, a canonical model for the closure of $\varphi$ is an epistemic model with awareness. 
\end{lemma}
\begin{proof}
  We prove that a canonical model for the closure of $\varphi$ satisfies the definition of an epistemic model with awareness. For $W^*$ and $V^*$, the proofs are trivial.
  \begin{itemize}
  \item For $(\mathop{\sim_i})^*$, it is proved in a way similar to \textbf{S5}. For the reflexivity, it is sufficient to prove $(w,w)\in (\mathop{\sim_i})^*$ for all $w\in W^*$. This follows from $\mathrm{T_I}$. For the symmetricity, it is sufficient to prove that for all $w,v\in W$, if $(w,v)\in (\mathop{\sim_i})^*$, then $(v,w)\in (\mathop{\sim_i})^*$. This follows from $\mathrm{T_I}$ and $\mathrm{5_I}$. For the transitivity, it is sufficient to prove that, for all $w,v,u\in W^*$, if $(w,v)\in (\mathop{\sim_i})^*$ and $(v,u)\in (\mathop{\sim_i})^*$, then $(w,u)\in (\mathop{\sim_i})^*$. This follows through $I_i\varphi\to I_i I_i\varphi$ from $\mathrm{K_I, T_I, 5_I}$, and $\mathrm{GI}$.
  \item For $(\mathscr{A}_i)^*$, it is sufficient to prove that 
  for all $w\in W^*$, if $(w,v)\in (\mathop{\sim_i})^*$ then $(\mathscr{A}_i)^*(w) = (\mathscr{A}_i)^*(v)$. Suppose that $(w,v)\in (\mathop{\sim_i})^*$ and $p\in (\mathscr{A}_i)^*(w)$. It follows from the definition of $(\mathscr{A}_i)^*(w)$ and $\mathrm{IA}$ that $I_i A_i p$. From the definition of $(\mathop{\sim_i})^*$, $A_i p\in v$. Thus, $p\in(\mathscr{A}_i)^*(v)$. The converse direction can be proved from $\mathrm{INA}$.
\end{itemize} 
\noindent
\end{proof}

\begin{lemma}
  Relation $(\mathop{\approx_i})^*$ on each canonical model is an A-equivalence relation on an epistemic model with awareness.
\end{lemma}
\begin{proof}
  First, we prove that $(\mathop{\approx_i})^*$ is an equivalence relation. It can be proved in the same way as $(\mathop{\sim_i})^*$ because both operators defining each relation have the axioms corresponding to $\mathrm{K, T},$ and $\mathrm{5}$ axioms. Next, we prove that for all $(w,v)\in (\mathop{\approx_i})^*$ and every $p\in (\mathscr{A}_i)^*(w)$, if $w\in V^*(p)$, then $v\in V^*(p)$ and vice versa. From left to right, suppose that $w\in V^*(p)$ for $p\in (\mathscr{A}_i)^*(w)$, then $p, A_i p\in w$. It follows that $p\in v$ from $\mathrm{AA[\approx]}$ and the definition of $(\mathop{\approx_i})^*$. The converse direction can be proved by that $(\mathop{\approx_i})^*$ is an equivalence relation.
\end{proof}

We prove a corollary of Lemma 6.
\begin{lemma}
  Let $\Phi$ be the closure of a formula, $M^*$ the canonical model for $\Phi$, and $\Gamma$ a consistent set in $\Phi$. For every $\varphi\in\Phi$, $\Gamma\vdash\varphi$ iff $\varphi\in\Delta$ for every maximal consistent set $\Delta \in W^*$ such that $\Gamma\subseteq\Delta$.
\end{lemma}
\begin{proof}
  From left to right, suppose that $\Gamma\vdash\varphi$, $\Gamma\subseteq\Delta$, and $\Delta\in W^*$. There exists a finite subset $\Gamma'\subseteq\Gamma$ such that $\vdash \bigwedge\Gamma'\to \varphi$. Since $\Gamma'\subseteq\Delta$, $\varphi\in\Delta$. From right to left, we prove it by contraposition. It is sufficient to prove that if $\Gamma\nvdash\varphi$, then there exists a maximal consistent set $\Delta$ such that $\Gamma\subseteq\Delta$, $\Delta\in W^*$, and $\varphi\not\in\Delta$. By the assumption, $\Gamma\cup\{\neg\varphi\}$ is a consistent set in $\Phi$. Thus, we obtain a maximal consistent set $\Delta$ in $\Phi$ such that $\Gamma\subseteq\Delta$ and $\varphi\not\in\Delta$ by Lemma 6.
\end{proof}

The next task is to define a $[\circ^+]_i$-path and $\varphi$-path on the canonical model $M^*$ for a closure. 
\begin{df}
  Let $\Phi$ be the closure of a formula and $M^*$ the canonical model for $\Phi$. A $[\circ^+]_i$-path from $\Gamma$ is a sequence $\Gamma_0,\dots,\Gamma_n$ of maximal consistent sets in $W^*$ such that $(\Gamma_k,\Gamma_{k+1}) \in (\mathop{\sim_i})^*\circ (\mathop{\approx_i})^*$ for all $k$, where $0\leq k\leq n$ and $\Gamma_0 = \Gamma$. The length of $\Gamma_0,\dots,\Gamma_n$ is $n$. A $\varphi$-path is a sequence $\Gamma_0,\dots,\Gamma_n$ of maximal consistent sets in $W^*$ such that $\varphi\in\Gamma_k$ for all $k$, where $0\leq k\leq n$.
\end{df}  

\begin{lemma}
  Let $\Phi$ be the closure of a formula, $M^*$ the canonical model for $\Phi$, and $\Gamma$ and $\Delta$ maximal consistent sets in $W^*$. If $\bigwedge\Gamma \wedge \neg[\approx]_i I_i \neg \bigwedge\Delta$ is consistent, then $(\Gamma,\Delta)\in (\mathop{\sim_i})^*\circ (\mathop{\approx_i})^*$ in the canonical model $M^*$. 
\end{lemma}

\begin{proof}
  It is sufficient to prove that if $\bigwedge\Gamma \wedge \neg[\approx]_i I_i \neg \bigwedge\Delta$ is consistent and $[\approx]_i I_i\varphi \in \Gamma$, $\varphi\in\Delta$. We prove it by the contradiction. From the assumptions, $[\approx]_i I_i\varphi \wedge \neg[\approx]_i I_i\neg \bigwedge\Delta$ is consistent. Suppose that $\varphi\not\in\Delta$, then $\neg\varphi\in\Delta$. It follows that $[\approx]_i I_i\varphi \wedge\neg[\approx]_i I_i\varphi$ is consistent. This is a contradiction. Thus, $\varphi\in\Delta$. 
\end{proof}

\begin{lemma}
  Let $\Phi$ be the closure of a formula and $M^*$ the canonical model for $\Phi$. If $[\circ^+]_i\varphi\in\Phi$ and $\Gamma \in W^*$, then $[\circ^+]_i\varphi\in\Gamma$ iff every $[\circ^+]_i$-path from $\Gamma$ is a $\varphi$-path and $[\circ^+]_i\varphi$-path. 
\end{lemma}

\begin{proof}
  ($\Rightarrow$) We prove it by induction on the length of a $[\circ^+]_i$-path.
  \begin{itemize}
    \item For the base case, suppose that the length of a $[\circ^+]_i$-path is $0$, $[\circ^+]_i \varphi\in\Phi$, and $[\circ^+]_i\varphi\in\Gamma$. $\Gamma =\Gamma_0 = \Gamma_n$. By $\mathrm{MIX}$, $\varphi\in\Gamma$.
    \item For induction steps, suppose that the length of a $[\circ^+]_i$-path is $k+1$, $[\circ^+]_i\varphi\in\Phi$, and $[\circ^+]_i\varphi\in\Gamma$. By the induction hypothesis, $[\circ^+]_i\varphi\in\Gamma_k$. It follows that $\varphi, [\circ^+]_i\varphi\in \Gamma_{k+1}$ from $\mathrm{MIX}$ and the definitions of $(\mathop{\sim_i})^*$ and $(\mathop{\approx_i})^*$. 
  \end{itemize}
  ($\Leftarrow$) Let $S([\circ^+]_i,\varphi)$ be a set of maximal consistent sets $\Delta$ in $W^*$ such that every $[\circ^+]_i$-path from $\Delta$ is a $\varphi$-path. We introduce the following formula: \[\chi = \bigvee_{\Delta\in S([\circ^+]_i,\varphi)}\bigwedge \Delta.\] 
  Suppose that every $[\circ^+]_i$-path from $\Gamma$ is a $\varphi$-path. As the first step, we prove the following: 
  \begin{align*}
  (1)\ \vdash &\bigwedge\Gamma\to \chi; 
  \quad (2)\ \vdash \chi\to\varphi;
  \quad (3)\ \vdash \chi\to [\approx]_i I_i\chi. 
  \end{align*}
  \begin{itemize}
  \item For (1), since $\Gamma\in S([\circ^+]_i,\varphi)$, $\vdash \bigwedge\Gamma\to \chi$.
  \item For (2), $\varphi\in\Delta$ for every $\Delta \in S([\circ^+]_i,\varphi)$ since every $[\circ^+]_i$-path from $\Delta$ is a $\varphi$-path. Thus, $\varphi$ is derived from $\chi$. 
  \item For (3), we prove it by contradiction. Suppose $\chi\wedge\neg[\approx]_i I_i\chi$ is consistent, then there exists $\Delta$ such that $\bigwedge\Delta\wedge\neg[\approx]_i I_i\chi$ is consistent because of the construction of $\chi$. A $W^*$ can be regarded as the whole set of combinations satisfying a condition of formulae in $\Phi$, and $\chi$ as a representation of a particular set of the combinations. The formula $\neg\bigvee_{\Lambda\in W^*\setminus S([\circ^+]_i,\varphi)} \bigwedge\Lambda$ is equivalent to $\chi$ because a particular set of combinations represented by $\chi$ is the 
  complement of the set containing all the other combinations. Therefore, $\bigwedge\Delta\wedge\neg[\approx]_i I_i \neg\bigvee_{\Lambda\in W^*\setminus S([\circ^+]_i,\varphi)} \bigwedge\Lambda$ is consistent. There exists $\Lambda$ such that $\bigwedge\Delta\wedge\neg[\approx]_i I_i \neg\bigwedge\Lambda$ is consistent. By Lemma 10, $(\Delta,\Lambda)\in (\mathop{\sim_i})^*\circ(\mathop{\approx_i})^*$ and there exists a $[\circ^+]_i$-path from $\Delta$ but not a $\varphi$-path. This is a contradiction. Thus, $\vdash \chi\to [\approx]_i I_i\chi$.
  \end{itemize}
  By $(3)$ and $\mathrm{G[\circ^+]}$, $\vdash [\circ^+]_i(\chi\to [\approx]_i I_i\chi)$. It follows that $\vdash\chi\to [\circ^+]_i \chi$ from $\mathrm{IND}$. By $(1)$, $\vdash\bigwedge \Gamma\to [\circ^+]_i\chi$. By $(2)$, $G[\circ^+]$, and $K_{[\circ^+]}$, $\vdash\bigwedge \Gamma\to [\circ^+]_i\varphi$. Thus, $[\circ^+]_i\varphi\in\Gamma$.
\end{proof}
  
We proceed to prove the truth lemma: in a canonical model, a true formula at a world is included in that world.
\begin{lemma}
  Let $\Phi$ be the closure of a formula and $M^*$ the canonical model for $\Phi$. For all $w\in W^*$ and every $\varphi\in \Phi$, $M^*,w\vDash\varphi$ iff $\varphi\in w$.
\end{lemma}
\begin{proof}
  We prove this by induction on the structure of formulae. The proofs of the base case and the logical connectives are straightforward. 
\begin{itemize}
  \item For the case of $A_i \psi$, we prove it by induction on the structure of $\psi$.
  \begin{itemize}
  \item For the base case, suppose that $M^*,w\vDash A_i p$, then $p\in\{q\mid A_i q\in w\}$ and $A_i p\in w$. As for the converse direction, suppose that $A_i p \in w$. This means that $(\mathscr{A}_i)^*(w)$ contains $p$. Thus, $M^*,w\vDash A_i p$.
  \item For the other cases, we obtain the desired proof by induction hypothesis and decomposing the formula using corresponding axioms: $\mathrm{AN}$, $\mathrm{AC}$, $\mathrm{AA}$, $\mathrm{AI}$, $\mathrm{A[\approx]}$, $\mathrm{A[\circ^+]}$, and $\mathrm{AE}$. We demonstrate the proof for $A_i\neg\chi$.
    \begin{itemize}
      \item Suppose that $M^*,w\vDash A_i\neg\chi$, which means that $At(\neg \chi)\subseteq \{q\mid A_i q\in w\}$. Since $At(\neg \chi) = At(\chi)$, $M^*,w\vDash A_i\chi$. It follows that $A_i\neg \chi\in w$ from the induction hypothesis and $\mathrm{AN}$. The converse direction can be proved similarly.
    \end{itemize}
  \end{itemize}
  \item For the case of $I_i \psi$ and $[\approx]_i\psi$, it is proved in a strategy similar to \textbf{S5}. We only demonstrate the case of $I_i\psi$ here. From left to right, suppose that $M^*,w\vDash I_i\psi$, which means that for all $v$ such that $\{\chi\mid I_i\chi\in w\}\subseteq v$, $M^*,v\vDash \psi$. It follows from the induction hypothesis that for all $v$ including $w$ itself, $\psi\in v$. By Lemma 9, $\{\chi\mid I_i\chi\in w\}\vdash\psi$. There exists a finite subset $\{\chi_1,\dots,\chi_n\}$ of $\{\chi\mid I_i\chi\in w\}$ such that $\vdash (\chi_1\wedge\dots\wedge \chi_n)\to \varphi$. It follows that $\vdash (I_i\chi_1\wedge\dots\wedge I_i\chi_n) \to I_i\varphi$ from $\mathrm{GI}$ and $\mathrm{K_I}$. Since a set $\{I_i\chi_1,\dots,I_i\chi_n\}$ is included in $w$, $w\vdash I_i\psi$. Thus, $I_i\psi\in w$. The converse direction is straightforward.
  \item For the case of $[\circ^+]_i\psi$, suppose that $M^*,w\vDash [\circ^+]_i\psi$, which means that for all $v$ such that $(w,v)\in ((\mathop{\sim_i})^*\circ (\mathop{\approx_i})^*)^+$, $M^*,v\vDash \psi$. This means that every $[\circ^+]_i$-path from $w$ is a $\psi$-path and $[\circ^+]_i\psi$-path. By Lemma 11, $[\circ^+]_i\psi\in w$. Similarly, the converse direction can be proved.
  \item For the case of $E_i\psi$, it is proved by decomposing it into $A_i\psi$ and $[\circ^+]_i\psi$. Suppose that $M^*,w\vDash E_i\psi$, then $M^*,w\vDash A_i\psi$ and $M^*,w\vDash [\circ^+]_i\psi$. It follows from the proofs in terms of $A_i$ and $[\circ^+]_i$ that $A_i\psi\in w$ and $[\circ^+]_i\psi\in w$. By $\mathrm{EA[\circ^+]}$, $E_i\psi\in w$. Similarly, the converse direction can be proved.
\end{itemize}
\noindent
\end{proof}

\begin{lemma}
    Let $\Phi$ be the closure of a formula and $\Gamma$ a maximal consistent set in $\Phi$. For every $\varphi\in\Phi$, if $\varphi\in \Gamma$ for every maximal consistent set $\Gamma$ in $\Phi$, then $\vdash\varphi$.
\end{lemma}
\begin{proof}
  We prove it by contraposition. Suppose $\nvdash\varphi$, then $\{\neg \varphi\}$ is a consistent set in $\Phi$. By Lemma 6, there exists a maximal consistent set in $\Phi$ such that the set does not contain $\varphi$.
\end{proof}

Finally, we prove the completeness theorem. 
\begin{thm}
  For every $\varphi\in\mathcal{L}_{AIL}$, if $\vDash\varphi$, then $\vdash\varphi$.
\end{thm}
\begin{proof}
  Suppose that $\vDash\varphi$. For the canonical model $M^*$ for the closure of $\varphi$, $M^*\vDash\varphi$ by Lemma 7 and Lemma 8. Thus, for every maximal consistent set $w$ in $W^*$, $\varphi\in w$ by Lemma 12. Therefore, $\vdash\varphi$ by Lemma 13.
\end{proof}

\section{Related Work}
There have been studies based on a similar idea to this paper, which is to associate an agent's awareness with indistinguishability among possible worlds or states.
\paragraph*{Bisimulation}
In \cite{van2009awareness,ditmarsch2011becoming,van2018implicit}, the authors provided \textit{awareness bisimulation} and proposed the notion of \textit{speculative knowledge}. The bisimulation relation $\mathfrak{R}[\mathscr{A}(w)]$ is between the models that cannot be distinguished by formulae consisting only of aware atomic propositions. Speculative knowledge is defined in terms of the awareness bisimulation; formally, $\varphi$ is speculatively known ($K^S\varphi$) at world $w$ if and only if $\varphi$ is true in all worlds that are $\mathscr{A}(w)$ awareness bisimilar to some accessible worlds from $w$. The notion captures a true proposition in all accessible worlds for every possible interpretation of all propositions of which an agent is unaware. As with our logic, their logic thereby formalizes indistinguishability among possible worlds induced by a lack of awareness.

However, this knowledge differs from that referred to by any of our operators. At $w$ in $M$ illustrated in Figure 9, $K^S p$ is true. This is because the truth value of the aware proposition $p$ is the same in all worlds that are $\{p\}$ awareness bisimilar to $w$, such as $w'$. In contrast, $\neg[\circ^+]p$ and $\neg E p$ are true at $w$ in $M$.
The observed difference is attributed to how to capture the notion of indistinguishability dependent on awareness/unawareness, that is, whether to capture awareness-indistinguishability at the model level or the world level. The formal comparative analysis is left for future work.
\begin{figure}
  \centering
  \includegraphics[width=13cm]{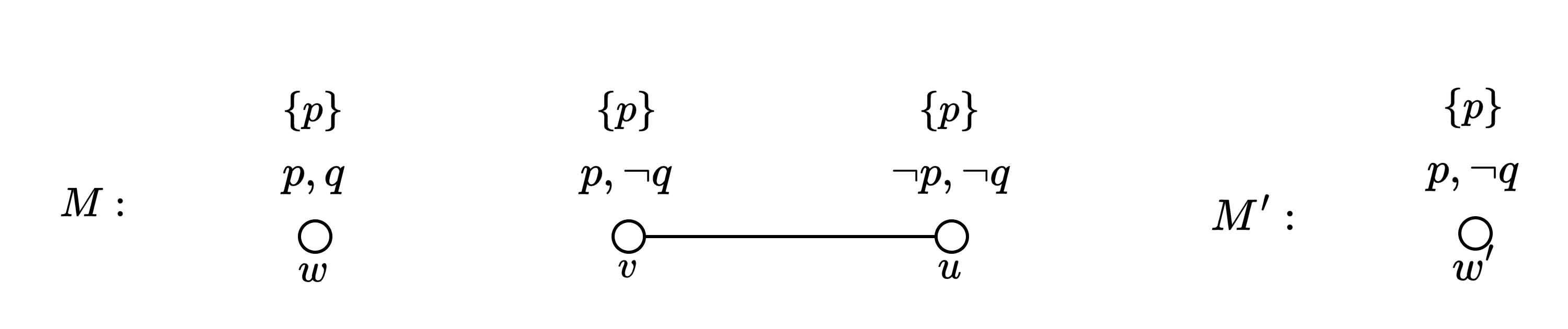}
  \caption{Possible world models with awareness for a single agent, where a solid line represents an IK-accessibility relation. Pointed model $(M,w)$ is $\mathscr{A}(w)$ awareness bisimilar to $(M',w')$. For readability, we omit the reflexivity.}
\end{figure}

\paragraph*{Equivalence relations among worlds}
In \cite{yudai2022-1}, the authors defined an equivalence relation expressing awareness-indistinguishability and, as with us, proposed a different definition of explicit knowledge from FH logic. Their approach inspires us to define our logic. They addressed a notion of awareness from an agent's perspective, that is, what ``agent $i$ is aware of a proposition from another agent $j$'s perspective'' is, representing it as $A^i_j \varphi$. We may construe the notion as information concerning another agent's awareness that an agent can use to obtain her explicit knowledge, and then represent it as explicit knowledge of awareness. Under this reading, the formula can be represented semantically within an epistemic model with awareness. Their logic does not handle a chain of perspectives with a depth of more than three, whereas our logic handles the chain with a finite but arbitrary depth. Furthermore, their semantics adopted a \textit{global definition} of awareness sets; hence, agents' awareness is implicitly known to every agent, that is, for all $w\in W$, $M,w\vDash I_i A^j_k$ for each $i,j,k\in\mathcal{G}$. In contrast, our models adopt a \textit{local definition}, which is one of the differences from their models.

In \cite{udatsu}, the authors also introduced an equivalence relation between worlds sharing the same truth value for aware propositions. Adopting the global definition of awareness, they proposed \textit{static abstraction operator} that refers to equivalence classes representing an `abstraction' of the given model, which operator is the global counterpart of $[\mathop{\approx}]$ in $\mathcal{AIL}$.

Besides these awareness logics, some studies introduced the same type of equivalence relations. The papers \cite{DBLP:conf/dagstuhl/Grossi09,grossi2015ceteris} proposed logics formalizing \textit{ceteris paribus} clauses, which come from the philosophy of science and are understood in the sense `all other things being equal.' The logics have the operator $[X]$ referring to a formula that holds in all worlds sharing the same valuation over a given set $X$ of atomic propositions. The operator works the same as $[\approx]_i$ from a formal aspect; $[\approx]_i\varphi = [\mathscr{A}_i]\varphi$. Furthermore, the authors developed a translation of the $[X]$ operator into a universal operator, embedding their logics into the class of Kripke models with universal relations. This technique can be applied to our logic, provided awareness sets are finite.

\paragraph*{Space-state model}
An \textit{unawareness structure}, proposed in \cite{heifetz2006interactive,heifetza2008canonical}, consists of a lattice of disjoint state-spaces with a partial order. Each space serves as a restriction on vocabulary instead of an awareness set. A state-space $S$ greater than $S'$ means that the states of $S$ describe situations with a richer vocabulary than those of $S'$. The states in a space associated with an agent's awareness (and spaces associated with other agents' awareness from the agent's viewpoint, for multi-agent cases) can be interpreted as a subjective description of an epistemic situation in the agent's mind. This structure (henceforth, an HMS model) provides a formalization of awareness and explicit knowledge in a syntax-free way \cite{DBLP:journals/logcom/BelardinelliR23}, having been applied to game theory \cite{heifetz2013unawareness,heifetz2013dynamic}.

We can view partitions by an A-equivalence relation for each agent as another formalization of their spaces in various vocabularies. For instance, let $p,q$ be atomic propositions and $\mathscr{A}_a(w) = \{p,q\}, \mathscr{A}_b(w) = \{p\}, \mathscr{A}_c(w) = \{q\},$ and $\mathscr{A}_d(w) = \emptyset$ for all $w$. For brevity, let us represent each world as a tuple of the literals; e.g., $(p,q)$ represents the world in which both propositions are true. Their lattice comprises four state-spaces: 
\begin{itemize}
 \item $\{(p,q),(p,\neg q),(\neg p, q),(\neg p,\neg q )\}$;
 \item $\{(p),(\neg p)\}$;
 \item $\{(q),(\neg q )\}$;
 \item $\emptyset$,
\end{itemize}
whereas our model has four types of partitions: 
\begin{itemize}
 \item $\{[(p,q)]_{\mathop{\approx_a}},[(p,\neg q)]_{\mathop{\approx_a}},[(\neg p, q)]_{\mathop{\approx_a}},[(\neg p,\neg q )]_{\mathop{\approx_a}}\}$;
 \item $\{[(p,q)]_{\mathop{\approx_b}},[(\neg p,q)]_{\mathop{\approx_b}}\}$;
 \item $\{[(p,q)]_{\mathop{\approx_c}},[(p,\neg q)]_{\mathop{\approx_c}}\}$;
 \item $\{[(p,q)]_{\mathop{\approx_d}}\}$. 
\end{itemize}
\noindent
A formal comparison would require a translation between IK- and EK-accessibility relations and their possibility correspondences \cite{belardinelli2024implicit}. Nevertheless, it appears that our model forms a single model that overlaps each space.

The studies \cite{belardinelli2024implicit,DBLP:journals/logcom/BelardinelliR23} that explored the relation between HMS models and FH logic showed that the explicit knowledge was equivalent to the intersection of implicit knowledge and awareness, as with FH logic. However, given the structural similarity between an HMS model and our model, their model appears to provide rich structures enough to interpret our operators. Furthermore, the structures proposed in \cite{belardinelli2024implicit,DBLP:journals/logcom/BelardinelliR23} are composed of multiple Kripke models (the former: the category of Kripke models, the latter: a Kripke lattice model), having been designed to correspond to the lattice of spaces in an HMS model. A formal comparison with these logics elucidates the knowledge that HMS models and $\mathcal{AIL}$ capture.

\section{Conclusion}
In this paper, we proposed a notion of indistiguishability among possible worlds due to a lack of awareness; that is, awareness-indistinguishability, which the standard accessibility does not capture. We demonstrated, using a concrete example, that the definition of explicit knowledge in FH logic ($E\varphi\leftrightarrow I\varphi\wedge A\varphi$) may derive undesirable propositions that cannot be considered explicit knowledge, and that the proposed indistinguishability plays an important role in the representation of the knowledge. Therefore, we focused on this indistinguishability and revisited the definition of explicit knowledge, proposing the logic $\mathcal{AIL}$ for representing the knowledge that is a refinement of that in FH logic. The indistinguishability is formally expressed by an equivalence relation between worlds depending on an agent's awareness, called A-equivalence relation. In our semantics, the aware implicit knowledge does not necessarily entail explicit knowledge ($E\varphi\to I\varphi\wedge A\varphi$ but $I\varphi\wedge A\varphi\not\to E\varphi$), which enables us to rule out the undesirable propositions. We then proved that our logical language is more expressive than that of an FH logic (Theorem 1) and showed that the latter is embeddable into $\mathcal{AIL}$ (Theorem 2). Furthermore, we provided the axiomatic system $\mathbf{AIL}$ and proved its soundness and completeness (Theorems 3 and 4).

There are several directions for further research. 
\begin{description}
 \item[Dynamics] One important direction is to explore the information referred to by $[\mathop{\approx}]$ and $[\circ^+]$ more thoroughly. As we explored what kind of knowledge $I$ and $[\circ^+]$ capture with dynamic operators in Subsection 3.3, the introduction of them \cite{van2011logical,van2010dynamics,Fernandez-Fernandez2021-FERAIL-2} sharpens our understanding of these notions. Our definition of A-equivalence relation is model-dependent, which facilitates its extension into a logic involving dynamics. 
 \item[Group knowledge] Another direction is to extend our framework to group knowledge. Incorporating the notion of common knowledge is especially of interest. In strategic situations where a player makes a decision with knowledge of the other's action, it is valuable to represent common knowledge from the viewpoint of an outside observer with complete awareness. At the same time, it would be equally valuable to represent common knowledge from the viewpoint of an agent with incomplete awareness. Our framework could be applied to formalize such common knowledge.
\end{description}
Moreover, clarifying the relations with the logics mentioned in Section 6 would also be a promising strategy for gaining a deeper understanding of what those logics and our logic capture.

\section*{Acknowledgements}
We thank the anonymous reviewer, Nobu-Yuki Suzuki, and Hiromi Inada for their helpful comments. This work was financially supported by 23K21869.

\bibliographystyle{plain}
\bibliography{refs}

\end{document}